\documentclass[12pt]{article}
\usepackage{preamble}

\begin{document}

\def\Ber{\mathsf{Ber}}
\def\BE{\mathsf{BE}}
\def\hwt{\mathsf{wt}}
\def\sym{\mathsf{sym}}
\def\ends{\mathsf{ends}}
\def\middle{\mathsf{mid}}
\def\Bin{\mathsf{Bin}}
\def\Ber{\mathsf{Ber}}
\def\symdiff{\bigtriangleup}
\def\sK{\ol{K}}

\DeclarePairedDelimiterX{\infdivx}[2]{(}{)}{%
  #1\;\delimsize\|\;#2%
}
\newcommand{\KL}{D\infdivx}

\newcommand{\summ}[1]{%
  \begingroup\noexpandarg
  \StrLen{#1}[\temp]%
  1^\top
  \ifnum\temp>1
    (#1)%
  \else
    #1%
  \fi 
  \endgroup
}

\def\toggleComments{1}
\ifnum\toggleComments=1

\newcommand{\blue}[1]{{\color{blue} {#1}}}
\newcommand{\chin}[1]{\footnote{{\bf \color{cyan}Chin}: {#1}}}
\newcommand{\manu}[1]{\footnote{{\bf \color{purple}Manu}: {#1}}}
\newcommand{\peter}[1]{\footnote{{\bf \color{green}Peter}: {#1}}}
\newcommand{\harm}[1]{\footnote{{\bf \color{orange}Harm}: {#1}}}
\fi
\ifnum\toggleComments=0
\newcommand{\chin}[1]{}
\newcommand{\manu}[1]{}
\newcommand{\peter}[1]{}
\newcommand{\harm}[1]{}
\fi

\def\epsilon{\varepsilon}
\newcommand{\e}{\epsilon}

\title{Pseudorandomness, symmetry, smoothing: I}

\author{%
Harm Derksen\footnote{Partially supported by NSF grant DMS 2147769.} \\
Northeastern University\\
ha.derksen@northeastern.edu
\and
Peter Ivanov\footnote{Supported by NSF grant CCF-2114116.}\\
Northeastern University\\
ivanov.p@northeastern.edu
\and
Chin Ho Lee \\
North Carolina State University\\
chinho.lee@ncsu.edu	
\and
Emanuele Viola\footnote{Supported by NSF grant CCF-2114116.} \\
Northeastern University\\
viola@ccs.neu.edu
}

\maketitle

\begin{abstract}

We prove several new results about bounded uniform and small-bias distributions.  A main message is that, small-bias, even perturbed with noise, does not fool several classes of tests better than bounded uniformity.  We prove this for threshold tests, small-space algorithms, and small-depth circuits.  In particular, we obtain small-bias distributions that
\begin{itemize}
  \item achieve an optimal lower bound on their statistical distance to any bounded-uniform distribution.  This closes a line of research initiated by Alon, Goldreich, and Mansour in 2003, and improves on a result by O'Donnell and Zhao.
  \item have heavier tail mass than the uniform distribution.  This answers a question posed by several researchers including Bun and Steinke.
  \item rule out a popular paradigm for constructing pseudorandom generators, originating in a 1989 work by Ajtai and Wigderson.  This again answers a question raised by several researchers.  For branching programs, our result matches a bound by Forbes and Kelley.
\end{itemize}
Our small-bias distributions above are symmetric.  We show that the xor of any two symmetric small-bias distributions fools any bounded function.  Hence our examples cannot be extended to the xor of two small-bias distributions, another popular paradigm whose power remains unknown.
We also generalize and simplify the proof of a result of Bazzi. 
\end{abstract}

\thispagestyle{empty}
\newpage

\section{Introduction}
A distribution $D$ over $\{-1,1\}^n$ is \emph{$(\eps,k)$-biased} if for every $S \subseteq [n]$ of size $0 < |S| \le k$ we have $|\mathbb{E}[D^S]| \le \eps$, where $D^S := \prod_{i\in S} D_i$.  If $\eps = 0$ then any $k$ bits are uniform and $D$ is called \emph{$k$-wise uniform}; if $k = n$ then $D$ is called \emph{$\eps$-biased}.  The study of these distributions permeates and precedes theoretical computer science.  They were studied already in the 40's \cite{MR22821}, are closely related to universal hash functions \cite{CaW79}, error-correcting codes (see e.g.~\cite{DBLP:journals/eccc/HatamiH23}), and in their modern guise were introduced in the works \cite{ABI86,ChorGoHaFrRuSm85,NaN90}.

$(\eps,k)$-biased distributions behave like the uniform distribution in that several prominent tests cannot distinguish the two distributions.

\begin{definition} A test $f\colon \pmo^n \to [-1,1]$ is $\delta$-fooled by a distribution $D$ we have $|\mathbb{E}[f(U)] - \mathbb{E}[f(D)]|\le \delta$, where $U$ is the uniform distribution.
\end{definition}

At the same time, $(\eps,k)$-biased distributions can be sampled efficiently from a short seed. The combination of these facts enables many applications in algorithm design, coding theory, pseudorandomness, and more. For background we refer the reader to \cite{Vadhan12,DBLP:journals/eccc/HatamiH23,moti}, where bounds on seed lengths are also discussed.

To generate $k$-wise uniformity, seed length $ck\log n$ is sufficient, and necessary for $k<n^c$; while for $\e$-biased seed length $c\log(n/\e)$ is sufficient and necessary. In this paper, as in \cite{moti}, every occurrence of ``$c$" denotes a possibly different positive real number.  The notation ``$c_x$" for parameter(s) $x$ indicates that this number may depend on $x$ and only on $x$. Replacing ``$c$'' with $O(1)$ everywhere is consistent with one common interpretation of the big-Oh notation.

It is known that any $\e$-biased distribution is close to a $k$-wise uniform distribution in  total variation (a.k.a.~statistical, $L_1$, etc.) distance~\cite{AGM03,AlonAKMRX07,OZ18}.

\begin{lemma}[Theorem 1.1 \cite{OZ18}] \label{lemma:OZ} 
  Any $(\eps,k)$-biased distribution is $((\frac{e^3n}{k})^{k/2} \eps)$-close to a $k$-wise uniform distribution in total variation distance.
\end{lemma}

Hence, any property enjoyed by $k$-wise uniform distributions is inherited by distributions with bias $n^{-ck}$.  Unsurprisingly, that is precisely the bias for which the seed length of the latter matches that of the former, as discussed above.  The question arises as to which tests can be fooled with a larger bias, which would result in shorter seed length.

\begin{question} \label{question}
Which tests can distinguish some $\eps$-biased distribution from every $k$-wise uniform distribution, for an $\eps$ suitably larger than the bound in \Cref{lemma:OZ}?
\end{question}

For a concrete setting, one can think e.g.~$k = 10 \log n$ and $\eps = n^{-100}$, or any $\eps = n^{-o(\log n)}$.

\Cref{question} is a computational version of the classic question of the statistical distance between $\e$-biased and $k$-wise uniform distributions, studied in \cite{AGM03,AlonAKMRX07,OZ18}. \Cref{lemma:OZ} shows that for small $\e$, \emph{no test}, efficient or not, can distinguish the distributions.  Those works also give lower bounds in various ranges of parameters, which means that in those ranges, there exists some $\eps$-biased distribution such that every $k$-wise uniform distribution can be distinguished from it by some test.  However, the arguments in these papers either do not apply to the tests we consider below or for bias $\eps$ larger than $n^{-k}$, which is the regime of interest here.  These works are discussed more below.

A trivial test which cannot distinguish between small-bias and $k$-wise uniform distributions in the sense of \Cref{question} is parity.  By definition, the bias of parity is at most $\e$, which is $\e$-close to the bias of the uniform distribution, which is $0$.  And the uniform distribution is in particular $k$-wise uniform.

However, the answer to \Cref{question} was not known for various other classes of tests of interest.
To our knowledge, it was not known for symmetric or even threshold tests (a.k.a.~tail, deviation, concentration bounds, etc.).
In particular, Bun and Steinke posed the following question in \cite{BS15}.

\begin{quote}
  In this work, we focused on understanding the limits of $k$-wise independent distributions. Gopalan et al.~\cite{GopalanKM15} gave a much more sophisticated generator with nearly optimal seed length. But could simple, natural pseudorandom distributions, such as small-bias spaces, give strong tail bounds themselves?
\end{quote}
More concretely, the following question has been asked by several researchers.
We use $\summ x$ to denote the sum $\sum_{i=1}^n x_i$ of $x \in \pmo^n$, and $B$ to denote the binomial distribution $\summ U$.

\begin{question}\label{question:BS}
  Is it true that for every $a$, there exists $b$ such that $\Pr[\summ D \ge t] \le \Pr[B \ge t] + 1/n^a$ for every $n^{-b}$-biased distribution $D$ and $t = \sqrt{n \log n}$?
\end{question}

The answer to \Cref{question:BS} was known to be negative for $t=0$ (corresponding to the majority function): One can take $D$ to be uniform on strings of weight 0 modulo $3$, see \cite{Bazzi-pseudobinomial}.  The answer was also known to be positive when $t > n^{1/2+\eps}$ for a constant $\eps$ because all the relevant quantities are small enough; formally combine \Cref{cor:tail-bound} with \Cref{lemma:OZ}.  But for other values of $t$ closer to $\sqrt{n}$ the answer was less clear.

\paragraph{Smoothed tests and distributions.}
A main focus of this paper is on \emph{smoothed tests} and \emph{smoothed distributions}, which are tests and distributions perturbed by \emph{noise}.

\begin{definition}
$N_\rho$ is the noise distribution on $\pmo^n$, where each bit is independently set to uniform with probability $1-\rho$ and $1$ otherwise.  We write $D \cdot N_\rho$ for the coordinate-wise product of $D$ and $N_\rho$, which corresponds to bit-wise xor over $\{0,1\}$.
Note $x \cdot N_1 = x$ and $x \cdot N_0 = U$, for any $x$.

\end{definition}

For a test $f$ and a distribution $D$, a \emph{smoothed distribution} is defined as $D \cdot N_\rho$ and a \emph{smoothed test} is defined as $T_\rho f (x) := \E[f(x \cdot N_\rho)]$, for some retention rate $\rho \in [0,1]$.
Note that $\E[f(D\cdot N_\rho)] = \E[T_\rho f(D)]$ and we will use both viewpoints interchangeably throughout.

Note that smoothing does not increase the distance of any two distributions, with respect to any class of tests which is closed under shifts.  So distinguishing smoothed distributions is at least as hard as distinguishing the corresponding (non-smooth) distributions.
A main motivation for considering smoothed tests and distributions comes from several paradigms for constructing pseudorandom generators (PRGs) that combine $(\eps,k)$-biased distributions in different
ways.  These paradigms have been proposed in the last 15 years or so and are discussed next; for additional background, we refer the readers to the recent monograph \cite{DBLP:journals/eccc/HatamiH23}.

\paragraph{Small-bias plus (pseudorandom) noise.}
This paradigm goes back to Ajtai and Wigderson \cite{AjtaiW89}, but saw no further work until it was revived by Gopalan, Meka, Reingold, Trevisan, and Vadhan \cite{GopalanMRTV12}. It has been used in a number
of subsequent works including \cite{GavinskyLS12,ReingoldSV13,SVW17,HLV-bipnfp,LeeV-rop,CHRT18,ForbesKelley-2018,MRT19,Lee19,DBLP:conf/coco/DoronHH20,CLTW23}.
In particular, this paradigm gave rise to PRGs with near-optimal seed lengths for several well-studied classes of tests, including combinatorial rectangles~\cite{GopalanMRTV12,Lee19} and read-once AC$^0$ formulas~\cite{DHH19,DBLP:conf/coco/DoronHH20}.

The Ajtai--Wigderson paradigm comprises several steps.
A main step in this paradigm requires fooling (the average of) a random restriction of tests with a pseudorandom distribution.
(This can also be viewed as constructing a fractional pseudorandom generator~\cite{CHHL19,CHLT19,CGLSS21}.)
The works by Haramaty, Lee, and Viola \cite{LV-sum,HLV-bipnfp,LeeV-rop,Lee19}
have reinterpreted the notion of ``random restrictions'' as \emph{perturbing or xor-ing a small-bias distribution with noise}.
The perspective of noise has proved influential and is maintained in several following works, including the present one.

This perspective of noise has been used to prove a variety of new results in areas ranging from communication complexity \cite{HLV-bipnfp}, coding theory \cite{HLV-bipnfp,DBLP:conf/focs/SilbakKS19}, Turing machines \cite{viola-tm}, and one-way small-space computation \cite{ForbesKelley-2018,MRT19}.

In particular, building on the proof in \cite{HLV-bipnfp}, Forbes and Kelley \cite{ForbesKelley-2018}
significantly improved the parameters in \cite{HLV-bipnfp} and obtained pseudorandom generators with seed length $c\log^3 n$ that fool one-way logspace computation.  The main new feature of their result over the classic generator by Nisan \cite{Nis92} is that the order in which the input is read by the computation is arbitrary.  A main step in the result in \cite{ForbesKelley-2018} is showing that $c\log n$-wise uniformity xor-ed with noise fools logspace.  After their work, a natural question, asked independently by several researchers, is whether one can improve the seed length to $o(\log^3 n)$ by replacing $c \log n$-wise uniformity with polynomial bias.

\begin{question} \label{question:space}
  Does $1/\poly(n)$-bias plus noise fool one-way logspace? 
\end{question}

A positive answer would give improved generators for small-space
algorithms from $c\log^3 n $ to $c\log^2 n$, bringing the parameters of the result in \cite{ForbesKelley-2018},
which works in any order, in alignment with the classic fixed-order result of Nisan \cite{Nis92}.

In fact, the answer to \Cref{question:space} was not known even for the special case of one-way logspace algorithms which compute \emph{symmetric tests}; 
or for the even more restricted class of \emph{threshold tests} that have the form $\Id(\summ{x} \ge t)$, for any bias larger than $n^{-\log n}$.

\paragraph{Xor-ing small-bias distributions.}
Starting with \cite{BoV-gen}, researchers
have considered the bit-wise xor of several independent
copies of small-bias distributions.  The work \cite{LV-sum} draws a connection with the previous paradigm, showing that for a \emph{special class} of small-bias distributions, the paradigms are equivalent.

These distributions -- the xor of several small-bias distributions -- appear to
be significantly more powerful than a single small-bias distribution, while retaining
a modest seed length.  We refer to
\cite{ODonnell14, corr-survey, DBLP:journals/eccc/HatamiH23, moti} for background.

Despite several attempts \cite{BDVY08,MZ09,LV-sum}, no definitive counterexample to this paradigm has been bound; its power remains unknown.

\subsection{Our results}

In this work we prove several new results on $(\e,k)$-biased distributions.  A main message is that, for several natural classes of tests, \emph{small-bias distributions are no better than bounded uniformity}, i.e., we provide new information about \Cref{question}, and answer \Cref{question:BS,question:space}.

To set the stage, we start with showing that $k$-wise uniformity plus noise does fool symmetric functions with error $2^{-ck}$.  Note that noise is necessary, for parity is not fooled even by $(n-1)$-wise uniformity.  And even for threshold tests, the error would be polynomial \cite{DGJSV-bifh,benjamini2012kwise} rather than exponential in $k$.

\begin{theorem} \label{thm:bipnvssym}
Let $D$ be a distribution on $\pmo^n$ that is either
\begin{enumerate}[label=(\roman*),itemsep=2pt,parsep=0pt,topsep=2pt]
  \item \label{item:bipnvssym} $(2k)$-wise uniform, or
  \item \label{item:sbpnvssym} $(ck/n)^{4k}$-biased.
\end{enumerate}
Let $f\colon\pmo^n \to [-1,1]$ be symmetric.  Then $\abs{\E[f(U)] - \E[f(D \cdot N_\rho)] } \le c \cdot (e\rho)^{k/2}$.
\end{theorem}

\Cref{thm:bipnvssym}.\ref{item:bipnvssym} follows from \cite{ForbesKelley-2018} when $k\ge c\log n$, but their proof does not apply to smaller $k$.  Our result applies to any $k$, and this will be critical.

\Cref{thm:bipnvssym}.\ref{item:sbpnvssym} follows from \Cref{thm:bipnvssym}.\ref{item:bipnvssym} via the following simple extension of \Cref{lemma:OZ}, which we establish by taking noise into account.  (A direct application of \Cref{lemma:OZ} would give a larger error of $2^{-ck}$.)

\begin{lemma} \label{lemma:OZ-noise} 
  Let $D$ be an $(\eps,k)$-biased distribution on $\pmo^n$.
  Then $D \cdot N_\rho$ is $((\frac{e^3\rho n}{k})^{k/2} \eps)$-close to a $k$-wise uniform distribution in total variation distance.
\end{lemma}

A natural question is whether larger bias suffices in \Cref{thm:bipnvssym}.\ref{item:sbpnvssym}.  A main result in this work is that it does not, even for threshold tests.  The best possible bound for small-bias distributions is in fact obtained by combining \Cref{thm:bipnvssym}.\ref{item:bipnvssym} with the generic \Cref{lemma:OZ-noise}.  

\begin{theorem} \label{thm:sbpn-vs-uniform}
  There exists a $(ck/n)^k$-biased distribution $D$ such that $\Pr[\summ{D \cdot N_\rho} \ge 2\sqrt{kn}] \ge \Pr[ B \ge 2\sqrt{kn}] + (c\rho)^{2k}$ for every $\rho \in [0,1]$.
\end{theorem}

In fact, the distribution $D$ in \Cref{thm:sbpn-vs-uniform} (and 
in \Cref{thm:sbpn-vs-k-wise} below) is simultaneously $(2k-1)$-wise uniform.

\Cref{thm:sbpn-vs-uniform} gives a negative answer to \Cref{question:BS}.
Specifically, setting $\rho$ to be a constant and $k = \log n$ we obtain bias $1/n^{\omega(1)}$ but the error is $\ge 1/n^c$.

Note that our negative answer holds even with noise, while an answer to \Cref{question:BS} was not known even for plain small-bias distributions. This makes our results stronger.  Moreover, we do not know of a simpler proof if one does not care about noise.  Indeed, we obtained several different proofs of essentially \Cref{thm:sbpn-vs-uniform}, see \cite{SSS-II}.  In all these proofs (including the one presented here) the small-bias distribution $D$ itself can be written as $D := D' \cdot N_{c \rho}$, that is, by adding noise to another distribution.  Further adding noise to $D$ then comes at little cost, as already pointed out in \cite{LV-sum}, see \Cref{claim:noise-invariant}.  We also mention that some of these proofs cover wider range of parameters, and provide new information even for bounded uniformity.  We refer to \cite{SSS-II} for more on this.

Combining \Cref{thm:sbpn-vs-uniform} with \Cref{thm:bipnvssym}, one immediately obtains a \emph{smoothed} threshold test which distinguishes some $n^{-k}$-bias distribution from any $ck$-wise uniform distribution, answering \Cref{question} for such tests.

For general symmetric tests and the same distribution $D$, we prove a stronger result improving on the classic line of works in \cite{AGM03,AlonAKMRX07,OZ18} and finally 
matching \Cref{lemma:OZ-noise}.

\begin{theorem} \label{thm:sbpn-vs-k-wise}
  There exists a $(ck/n)^k$-biased distribution $D$ such that 
  for every $\rho \in [0,1]$ the following holds.
  There exists a symmetric function $f\colon\pmo^n \to \zo$ such that for every $(2k)$-wise uniform distribution $D_{2k}$,
  \[
    \E[ f(\summ{D \cdot N_\rho}) ] \ge \E[f(D_{2k})] + \Bigl(\frac{c\rho}{\sqrt{\log(1/2\rho)}}\Bigr)^{2k}.
  \]
\end{theorem}
Again, this result was not known even without noise.  Note that \Cref{thm:sbpn-vs-k-wise} implies the same separation without noise simply setting $\rho := 1$.  But the other way around is not clear.

An interesting question is whether one could prove a single result that implies both \Cref{thm:sbpn-vs-k-wise} and \Cref{thm:sbpn-vs-uniform}.

From \Cref{thm:sbpn-vs-uniform}, we derive several consequences on small-space computation and small-depth circuits.

\paragraph{One-way small space.}
We give a negative answer to \Cref{question:space}.

\begin{corollary} \label{cor:thr-thr-const-advant} For any $\rho\in (0,1]$, there is a distribution $D$ on $\{-1,1\}^{n}$
  that is $n^{-c\log_{1/\rho} n}$-biased and a threshold-of-thresholds
$T\colon\{-1,1\}^{n}\to\zo$ such that $\E[T(U)]-\E[T(D \cdot N_\rho)] \geq 1/3$.
In particular, there is a read-once branching program $T$ of width $n^c$ for which the inequality holds.
\end{corollary}

\Cref{cor:thr-thr-const-advant}, in combination with \Cref{lemma:OZ-noise}, matches a result in \cite{ForbesKelley-2018}, which shows that the error is $\rho^{ck}$ for $k$-wise uniform when $k \ge c\log_{1/\rho} n$.

\begin{proof}[Proof of \Cref{cor:thr-thr-const-advant} from \Cref{thm:sbpn-vs-uniform}]
We divide the input into $\sqrt{n}$ blocks, and in each block sample an independent copy of the $n^{-k/2}$-biased
distribution from \Cref{thm:sbpn-vs-uniform} on $\sqrt{n}$
bits.  The resulting distribution has the required properties, since the bias of a test that spans multiple blocks equals the product of the biases in each block.

In each block, a suitable threshold tells $D \cdot N_\rho$ from uniform with advantage $(c\rho)^k \geq n^{-0.1}$
for $k=c\log_{1/\rho} n$.  A threshold of $\sqrt{n}$
such blocks is sufficient to boost the advantage to constant.

Finally, this threshold-of-thresholds computation can be implemented with $c \log n$ bits of space, by simply maintaining two counters.
\end{proof}

What may have made this problem harder is that it was not clear what
distinguishing bound one should expect in \Cref{thm:sbpn-vs-uniform}.
One may be tempted to aim for larger advantage, perhaps independent from $k$. But as we showed in \Cref{thm:bipnvssym},
this is false: $k$-wise uniformity plus noise fools thresholds
with error $2^{-ck}$. One can then ask if $k$-wise uniformity 
fools with error $2^{-ck}$ more general classes of tests, like threshold
of thresholds. \Cref{cor:thr-thr-const-advant} shows this is also false.

\paragraph{Constant-depth circuits.}
Next we discuss a negative result for fooling the circuit class AC$^{0}$.  It is known that polylogarithmic independence or quasi-polynomial bias fools AC$^{0}$~\cite{Bazzi07,Razborov09,Braverman10,Tal17}, and these bounds are nearly tight.  But despite attempts \cite{LV-sum} it was not known if logarithmic uniformity plus noise, or polynomial bias plus noise suffices. We show that bias $n^{-\omega(1)}$ is necessary.

\begin{corollary} \label{cor:ac0-const-advant} For any $\rho\in(0,1]$ there is a distribution $D$ on $\pmo^n$
that is $n^{-c_\rho \log\log n}$ biased and an \textnormal{AC}$^0$ circuit $C$ of size $n^c$ and depth $c$ such that $\E[C(U)]-\E[C(D \cdot N_\rho)] \geq 1/3$.
\end{corollary}

The proof is similar to before, except we take blocks of polylogarithmic length, and set $k = c_\rho \log \log n$.  The threshold in each block can be computed in AC$^0$ since it's only on polylogarithmic number of bits.  By our setting of $k$, the advantage in each block is polylogarithmic, and so computing approximate majority \cite{Ajt83} (cf.~\cite{ViolaBPvsE}) suffices to have constant advantage.

%
%
%
%
%
%
%
%
%

\paragraph{Sum of small-bias distributions.}
%
A next natural question is whether our counterexamples can be extended to the xor of two small-bias distributions.  We show that they cannot.  Specifically, our small-bias distributions are symmetric, and we show that the sum of two such distributions fools any function (symmetric or not).

\begin{theorem} \label{thm:sym-xor-sym-fools-any}
Let $D_1$ and $D_2$ be two independent $n^{-20k}$-biased, symmetric distributions on $\pmo^n$.
Then $\abs{\E[f(D_1 \cdot D_2)] - \E[f]} \le c_k(n^{-0.3k})$ for any function $f\colon\pmo^n \to [-1,1]$.
\end{theorem}

In fact, we prove stronger results. We show that to fool any symmetric function it suffices for one of the two distributions to be symmetric (\Cref{cor:shifted-sym-bias-special}).  In fact, this holds even if one of the two distributions is any fixed string $x$ with $\summ{x} \le n^{0.99}$ (\Cref{thm:shifted-sym-bias}); and we complement this with a result showing that the result is false if $\summ{x}$ is large.
This is in \Cref{sec:shifted-sym-small-bias}.

\paragraph{Typical shifts.}
We generalize and simplify the proof of a result by Bazzi \cite{Bazzi15}.  We first discuss his result.  Let $C \subseteq \zo^n$ be a binary linear code with minimum distance $k+1$ and maximum distance $n-k-1$.
Let $U_{C^\perp}$ be the uniform distribution on the dual code of $C$, and $\bu \sim \zo^n$ be a uniform string.
Bazzi~\cite{Bazzi15} showed that for most shifts $\bu$, the distribution $\bu + U_{C^\perp}$ fools any symmetric function $f\colon\zo^n \to \zo$:
\[
  \E_{\bu}\bigl[ \abs[\big]{ \E[f(\bu + U_{C^\perp})] - \E[f] } \bigr] \le (k/n)^{ck} .
\]
It follows from the distance properties of $C$ that $U_{C^\perp}$ fools all parity tests of size at most $k$ and at least $n-k$ (with no error).  We show that in fact the conclusion above holds for every distribution $D$ that fools such parity tests, without requiring the distribution to be linear.

\begin{theorem} \label{thm:improve-Bazzi}
Let $D$ be a distribution on $\pmo^n$ such that $\E[D^S] = 0$ for every subset $S$ of size $\ell \in [1,k] \cup [n-k,n]$, and $\bu \sim \zo^n$ be a uniform string.
For every symmetric function $f\colon\pmo^n \to [-1,1]$,
  \[
    \E_{\bu}\bigl[ \abs[\big]{ \E[f(\bu \cdot D)] - \E[f] } \bigr]
    \le 6 (k/n)^{\frac{k-1}{4}} .
  \]
\end{theorem}

For context, we note that the condition on fooling large parity tests is necessary, as otherwise, one can consider the uniform distribution $D$ on strings with the parity $0$ (say), which is $(n-1)$-wise uniform, and for every shift $u$, the parity of $u+D$ (which is a symmetric function) is simply the parity of $u$.

Also, note that no fixed shift $\bu$ suffices, for else one can shift $D$ by this $\bu$ and give a counterexample.  This does not contradict the results discussed above about shifting symmetric small-bias distributions because the shift of a symmetric distribution is not in general symmetric.

\subsection{Krawtchouk polynomials}
All our results rely on bounds for the (shifted) Krawtchouk polynomials $\sK$, which can be defined by
\[
  \sK(k,t) := \sum_{\abs{S} = k} z^S ,
  \]
  where $z \in \pmo^n$ is any string such that $\summ z=t$, and $z^S$ is the product of the bits of $z$ indexed by $S$.
It can be shown that this is a degree-$k$ polynomial in $t$.

This is a classic quantity (cf. \cite{KL01}) and the bounds we need do not seem well known.

To illustrate the bounds we find it convenient to define the normalized version of $\sK$,
\[
  N \sK(k,t) := \frac{\sK(k,t)}{\binom{n}{k}}
\]
and its ``with replacement'' counterpart
\[
  NR(k,t):=\E_{f\colon[k]\to[n]}\Bigl[\tprod_{i\in[k]} z_{f(i)}\Bigr] = (t/n)^k,
\]
where $f$ is uniform.

Note that $N \sK$ is the same as $NR$ conditioned on $f$ having no collisions -- via the correspondence $S=\{f(i):i\in[k]\}$ -- which is the same as saying that the images of $f$ are picked from $[n]$ without replacement.

The bounds on $N \sK$ (and hence $\sK$) can now be understood as approximations to $NR(k,t)$.

First we prove a lower bound, needed for \Cref{thm:sbpn-vs-uniform}.  The proof is short and follows from known results on Krawtchouk polynomials.  However, we are unable to find the result we need in the literature.
\begin{claim} \label{claim:sK-lb}
  $N \sK(k,t) \ge (\frac{t}{2n})^k = NR(k,t)/2^k$ for $t \ge 2\sqrt{k(n-k)}$.
\end{claim}

For \Cref{thm:sbpn-vs-k-wise} we need an upper bound.  We could use a bound which to our knowledge appeared first in \cite{BIJLSV21}. Since the proof in the latter is somewhat technical, we also give a new simple proof of a stronger bound, stated next, building on the recent work by Tao \cite{Tao23-Maclaurin}.  We could also use \cite{BIJLSV21} for \Cref{thm:bipnvssym}, but we would get a bound of the form $(a\rho)^k$ for an unspecified constant $a$.  The stronger bound in \Cref{cor:bijlsv-ineq} proved here gives a better dependence on $\rho$ and gets us closer to the natural bound of $\rho^k$, which is currently not clear.

\begin{corollary} \label{cor:bijlsv-ineq}
For every $1 \le k \le n$, we have $
    \abs[\big]{N\sK(k,t)}
    \le ( \frac{k}{n} + \frac{t^2}{n^2} )^{k/2}.$
\end{corollary}

Note this is similar to $NR(k,t)$ except for the extra term $k/n$.

For other results we need additional bounds which hold in regimes where the above bounds are loose, such as when $k$ is close to $n/2$ and $t$ is close to $0$.  To illustrate, let $n$ be even and $\summ x = t := 0$, corresponding to $x \in \pmo^n$ being a balanced string.  Note that $\sK(k,0)$ is the $k$-th coefficient of the polynomial $(1-x^2)^{n/2}$, which is $(-1)^{k/2} \binom{n/2}{k/2} \Id(\text{$k$ is even})$.
In this case, \Cref{cor:bijlsv-ineq} gives an upper bound of $\binom{n}{k} (k/n)^{k/2}$.  In particular, when $k = n/2$, the bound is roughly $2^{3n/4}$.
By contrast, the bound given next by \Cref{prop:Harm-bound-H} is $2^{\frac{n}{2}H(k/n)}$, which when $k = n/2$ becomes $2^{n/2}$.

\begin{proposition} \label{prop:Harm-bound-H}
Let $k = \beta n$ and $t=(1-2\alpha) n$.
We have
$\log_2|\sK(k,t)| \leq \frac{n}{2} \bigl( 1 + H(\beta) - H(\alpha) \bigr),
$ where $H(\alpha)=-\alpha\log_2(\alpha)-(1-\alpha)\log_2(1-\alpha)$ is the binary entropy function.
\end{proposition}

A similar bound also appears in \cite[Lemma 2.1]{Pol19}.
Using the estimate $H(1/2 + \gamma) \ge 1 - 4\gamma^2$ for $\gamma \in [0,1/2]$, we have the following corollary.
\begin{corollary} \label{cor:Harm-bound}
  $\abs[\big]{\sK(k,t)} \le 2^{\frac{n}{2} (H(\frac{k}{n}) + \frac{t^2}{n^2}) }$.
\end{corollary}

In \Cref{sec:Kraw-bounds} we prove bounds more general than the above.

\section{Small-bias plus noise is far from bounded uniformity} \label{sec:sbpn-vs-uniform}

In this section we prove \Cref{thm:sbpn-vs-uniform,thm:sbpn-vs-k-wise}.  We build on the work by O’Donnell and Zhao \cite{OZ18}.  In particular, we use the same distribution $D$.  However, jumping ahead, our analyses differ from \cite{OZ18} in three ways, each of which is critical for us:
\begin{enumerate}[itemsep=0pt,parsep=2pt,topsep=3pt]
  \item while we analyze the same symmetric test in \Cref{thm:sbpn-vs-k-wise}, we use a new and explicit threshold test in \Cref{thm:sbpn-vs-uniform};
  \item the distinguishing advantages in \Cref{thm:sbpn-vs-uniform,thm:sbpn-vs-k-wise} are explicit and stronger.  This relies on our use of (and bounds for) Krawtchouk polynomials, instead of the Hermite approximation in \cite{OZ18};
  \item we take noise into account.
\end{enumerate}

We now define $D$ and derive some properties of it.  Then in the next subsections the theorems are proved in turn.

\begin{definition}
  For a parameter $\alpha \in [0,\frac{1}{5e}]$, define $D_\alpha \colon\pmo^n \to \R$ to be
  \[
    D_\alpha(x) := 2^{-n} \biggl(  1 + \alpha^k \binom{n}{2k}^{-\frac{1}{2}} \sum_{\abs{S}=2k} x^S \biggr) \quad \text{for every $x \in \pmo^n$.}
  \]
\end{definition}
Note that the right hand side is the Fourier transform of $D_\alpha$, and thus $2^n \hD_\alpha(\varnothing) = \sum_{x \in \pmo^n} D_\alpha(x) = 1$.
We now show that for $\alpha \le 1/(5e)$, we have $D_\alpha(x) \ge 0$ for every $x \in \pmo^n$ and thus it is a distribution.

\begin{claim} \label{claim:krawtchouk-even}
For $\alpha \le 1/(5e)$, we have $D_\alpha(x) \ge 0$ for every $x$.
\end{claim}

\begin{proof}
The key observation is that as a degree-$(2k)$ polynomial in $t$, the zeros of $\sK(2k, t)$ all lies within $\abs{t} \le 2\sqrt{(2k-1) (n-2k+2)} \le 2\sqrt{2kn}$~\cite{Levenshtein95} (see also \cite[Section~5]{KL01}).
As $2k$ is even, we know that when $x$ is the all-1 or all-(-1) string (i.e., $\summ x \in \{n,-n\}$), we have $\sK(2k, \summ x) := \sum_{\abs{S}=2k} x^S > 0$.
So $\sK(2k,\summ x)$ can only be negative when $\abs{\summ x} \le 2\sqrt{2nk}$.
In this interval, using \Cref{cor:bijlsv-ineq} and $\alpha \le 1/(5e)$, we have
\[
  \alpha^k \binom{n}{2k}^{-\frac{1}{2}} \abs[\bigg]{\sum_{\abs{S}=2k} x^S}
  \le \alpha^k \binom{n}{2k}^{-\frac{1}{2}} \binom{n}{2k} \left(\frac{10k}{n}\right)^k
  \le (5e\alpha)^k \le 1 . \qedhere \]
\end{proof}

By the above, $D_\alpha$ is a well-defined distribution whenever $\alpha \le 1/(5e)$.
The following claim is immediate.

\begin{claim}
For $\alpha \le 1/(5e)$, $D_\alpha$ is a distribution that is $(2k-1)$-wise uniform, $\alpha^k \binom{n}{2k}^{-1/2}$-biased, and $(\alpha e^3/2)^k$-close to $(2k)$-wise uniform.
\end{claim}
\begin{proof}
  The first two properties follow directly from the definition of $D_\alpha$, that is, for every nonempty $S$, we have $\abs{\E[D_\alpha^S]} = 2^n\abs{\wh{D_\alpha}(S)} = \alpha^k \binom{n}{2k}^{-1/2} \Id(\abs{S}=2k)$.
  The closeness to $(2k)$-wise uniform follows directly from \Cref{lemma:OZ}.
\end{proof}

Observe that the family $\{D_\alpha: \alpha \ge 0\}$ is closed under adding noise, as shown in the following claim.

\begin{claim} \label{claim:noise-invariant}
  $D_\alpha \cdot N_\rho = D_{\alpha \cdot \rho^2}$ for every $\rho \in [0,1]$. 
\end{claim}

\begin{proof}
  Observe that $N_\rho$ dampens each size-$(2k)$ (Fourier) coefficient of $D_\alpha$ by a factor of $\rho^{2k}$.
  To see this, note that $N_\rho(x_i) = \frac{1}{2}(1 + \rho x_i)$, and thus
  \[
    N_\rho(x) = 2^{-n} \Bigl( 1 + \sum_S \rho^{\abs{S}} x^S  \Bigr) .
  \]
  By Plancherel's theorem, each Fourier coefficient of the convolution $D_\alpha \cdot N_\rho$ is the product of the coefficient of $D_\alpha$ and $N_\rho$.
  So we have
  \[
    (D_\alpha \cdot N_\rho)(x)
    = 2^{-n} \biggl(  1 + \rho^{2k} \alpha^k \binom{n}{2k}^{-\frac{1}{2}} \sum_{\abs{S}=2k} x^S \biggr)
    =D_{\alpha \cdot \rho^2}(x) . \qedhere
  \]
\end{proof}

\subsection{Distinguishing $D_\alpha$ from uniform with a threshold}
We now show that a specific threshold distinguishes $D_\alpha$ from the uniform distribution.
First, we establish the following claim showing that $D_{\alpha}$ always puts more mass than $U$ on unbalanced strings.

\begin{claim} \label{claim:ptwise-lb}
  $\Pr[\summ D_\alpha = t] \ge \Pr[ B = t ] \cdot (1 + (\frac{\alpha t^2}{4kn})^k)$ for every $t \ge 2\sqrt{kn}$ and $\rho \in [0,1]$.
\end{claim}

\begin{proof}
  By our lower bound on Krawtchouk polynomials (\Cref{claim:sK-lb}), we have
  \begin{align*}
    \Pr[\summ D_\alpha = t]
    &= \Pr[ B = t ] \Bigl( 1 + \alpha^k \binom{n}{2k}^{-1/2} \sK(2k,t) \Bigr)  \\
    &\ge \Pr[ B = t ] \Bigl( 1 + \alpha^k \binom{n}{2k}^{1/2} \Bigl(\frac{t}{2n}\Bigr)^{2k} \Bigr)  \\
    &\ge \Pr[ B = t ] \Bigl( 1 + \Bigl(\frac{\alpha t^2}{4kn}\Bigr)^{k} \Bigr) . \qedhere
  \end{align*}
\end{proof}

\Cref{thm:sbpn-vs-uniform} then easily follows from \Cref{claim:ptwise-lb} by summing over all the points at the tail, and then setting $\alpha$ to be $\rho^2/(5e)$.
\begin{proof}[Proof of \Cref{thm:sbpn-vs-uniform}]
  From \Cref{claim:ptwise-lb}, it follows that
  \begin{align*}
    \Pr\Bigl[ \summ D_\alpha \ge 2\sqrt{kn} \Bigr]
    &\ge \sum_{t \ge 2\sqrt{kn}} \Pr\bigl[ B = t \bigr] \cdot \Bigl(1 + \Bigl(\frac{\alpha t^2}{4kn}\Bigr)^k \Bigr) \\
    &\ge \Pr\bigl[ B \ge 2\sqrt{kn} \bigr] \cdot \Bigl(1 + \Bigl(\frac{\alpha (2\sqrt{kn})^2}{4kn}\Bigr)^k \Bigr) \\
    &\ge \Pr\bigl[ B \ge 2\sqrt{kn}\bigr] + 2^{-ck} \cdot \alpha^k ,
  \end{align*}
  where the last inequality is because by tail bounds for the binomial distribution (cf. \cite{Ahle-book}) we have $\Pr[ B \ge 2\sqrt{kn}] \ge 2^{-ck}$.
  The theorem then follows from setting $\alpha$ to $\rho^2/(5e)$, and noting that $D_{1/(5e)} \cdot N_{\rho} = D_{\rho^2/(5e)}$ by \Cref{claim:noise-invariant}.
\end{proof}

\subsection{Distinguishing $D_\alpha$ from bounded uniformity with a symmetric test}

In this section, we prove \Cref{thm:sbpn-vs-k-wise}.
We start with a claim showing that it suffices to consider bounded symmetric functions instead of Boolean symmetric test.

\begin{claim}
    Let $D_1, D_2$ be any distributions on $\pmo^n$.
    Suppose there is a symmetric function $f\colon\pmo^n \to [-1,1]$ such that $\E[f(D_1)] \ge \E[f(D_2)] + \eps$.
    Then there exists a symmetric Boolean function $f'\colon \pmo^n \to \pmo$ such that $\E[f'(D_1)] \ge \E[f'(D_2)] + \eps$.
\end{claim}

\begin{proof}
Define $g\colon \{-n, \ldots, n\} \to [-1,1]$ so that $g(\summ x) := f(x)$.
Considering the randomized function $\bg\colon\{-n, \ldots, n\} \to \pmo$ defined by
\[
    \bg(w) := \begin{cases} 1 & \text{with probability $\frac{1 + g(w)}{2}$} \\
    -1 & \text{with probability $\frac{1 - g(w)}{2}$.} \end{cases}
\]
As $f$ is symmetric, we have
\[
    \E_{\bg}\Bigl[ \E\bigl[\bg(\summ D_1) \bigr] \Bigr]
    = \E[f(D_1)]
    \ge \E[f(D_2)] + \eps
    = \E_{\bg}\Bigl[ \E\bigl[\bg(\summ D_2) \bigr]\Bigr] + \eps ,
\]
and so by averaging, there must be a choice $g'$ of $\bg$ such that $\E[g'(\summ D_1)] \ge \E[g'(\summ D_2)] + \eps$.
Defining $f'\colon \pmo^n \to \pmo$ by $f'(x) := g'(\summ x)$ proves the claim.
\end{proof}

We now define our symmetric test.
For a sufficiently small constant $\alpha$, let $\beta := \frac{100}{\log(1/\alpha)}$.
Define the homogeneous degree-$k$ polynomial $p_\beta \colon\pmo^n \to \R$ by
  \[
    p_\beta(x) := \beta^k \binom{n}{2k}^{-\frac{1}{2}} \sum_{\abs{S}=2k} x^S = 2^n D_\beta(x) - 1.
  \]
  Let $f_\beta$ be its truncation so that it is bounded by $1$, that is, we define $f_\beta\colon\pmo^n \to [-1,1]$ by $f_\beta(x) := \min\{1, p_\beta(x)\}$.
  As $\alpha$ is sufficiently small, so is $\beta$.
  Thus, by \Cref{claim:krawtchouk-even}, we have $f_\beta(x) \ge -1$ and so $f_\beta(x) \in [-1,1]$ for every $x \in [-1,1]$.

  \begin{claim} \label{claim:OZ-distinguish}
    $\E[f_\beta(D_\alpha)] - \E[f_\beta(D)] \ge (\alpha \beta)^k/2$ for any $k$-wise uniform distribution $D$.
  \end{claim}

  \begin{proof}
    As $p_\beta$ has degree-$(2k)$, for any $(2k)$-wise uniform distribution $D$, we have $\E[f_\beta(D)] \le \E[p_\beta(D)] = \E[p_\beta(U)] = 0$.
    Note that we can write $f_\beta(x)$ as $p_\beta(x) - (p_\beta(x)-1)\Id(p_\beta(x) > 1)$, and so
    \begin{align}
      \E\bigl[f_\beta(D_\alpha)\bigr]
      &= \E\bigl[p_\beta(D_\alpha)\bigr] - \E\bigl[\bigl(p_\beta(D_\alpha)-1\bigr) \Id\bigl(p_\beta(D_\alpha) > 1\bigr) \bigr] \label{eq:OZ-goal} .
    \end{align}
    To bound $\E[f_\beta(D_\alpha)]$ from below, we will compute $\E[p_\beta(D_\alpha)]$ and then bound $\E[(p_\beta(D_\alpha)-1) \Id(p_\beta(D_\alpha) > 1)]$ from above.

    Observe that
    \begin{align*}
    \E\Bigl[ \tsum_{\abs{S}=2k} U^S \Bigr]
    &= \sum_{\abs{S}=2k} \E\bigl [U^S \bigr] = 0 \\
    \E\Bigl[ \bigl(\tsum_{\abs{S}=2k} U^S\bigr)^2 \Bigr]
    &= \sum_{\substack{\abs{S}=2k \\ \abs{T}=2k}} \E\bigl[U^{S \symdiff T}\bigr]
    = \binom{n}{2k} \\
    \E\Bigl[ \bigl(\tsum_{\abs{S}=2k} U^S\bigr)^3 \Bigr]
    &= \sum_{\substack{\abs{S}=2k\\ \abs{T}=2k\\ \abs{R}=2k}} \E\bigl[U^{S \symdiff T \symmdiff R}\bigr]
    = \binom{n}{2k}\binom{2k}{k}\binom{n-2k}{k} ,
    \end{align*}
    where the last equality is because the number of subsets $S,T,R \subseteq [n]$ of size $2k$ that satisfy $S \symdiff T = R$ is $\binom{n}{2k}\binom{2k}{k}\binom{n-2k}{k}$.
    
    We have
    \begin{align} \label{eq:OZ-1}
      \E\bigl[p_\beta(D_\alpha)\bigr]
      &= \sum_{x \in \pmo^n} D_\alpha(x) \E \bigl[ p_\beta(x) \bigr] \nonumber \\
      &= \sum_{x \in \pmo^n} 2^{-n} \biggl(1 + \alpha^k \binom{n}{2k}^{-\frac{1}{2}} \sum_{\abs{S}=2k} x^S \biggr) \biggl(  \beta^k \binom{n}{2k}^{-\frac{1}{2}} \sum_{\abs{S}=2k} x^S \biggr) \nonumber \\
      &= (\alpha\beta)^k \binom{n}{2k}^{-1} \E\Bigl[ \bigl(\tsum_{\abs{S}=2k}U^S\bigr)^2 \Bigr] \nonumber \\
      &= (\alpha\beta)^k \binom{n}{2k}^{-1} \binom{n}{2k} 
      = (\alpha\beta)^k .
    \end{align}
    We also have
    \begin{align} \label{eq:OZ-2}
      \E\bigl[p_\beta(D_\alpha)^2\bigr]
      &= 2^{-n} \sum_{x \in \pmo^n} \biggl( 1 + \alpha^k \binom{n}{2k}^{-\frac{1}{2}} \sum_{\abs{S}=2k} x^S \biggr) \biggr( \beta^k \binom{n}{2k}^{-\frac{1}{2}} \sum_{\abs{S}=2k} x^S \biggr)^2 \nonumber \\
      &= \beta^{2k} \binom{n}{2k}^{-1} \E\Bigl[ \bigl(\tsum_{\abs{S}=2k} U^S\bigr)^2 \Bigr] + (\alpha \beta^2)^k \binom{n}{2k}^{-\frac{3}{2}} \E\Bigl[ \bigl(\tsum_{\abs{S}=2k} U^S\bigr)^3\Bigr] \nonumber \\
      &= \beta^{2k} + (\alpha \beta^2)^k \binom{n}{2k}^{-\frac{1}{2}} \binom{2k}{k} \binom{n-2k}{k} \nonumber \\
      &\le \beta^{2k} + (\alpha \beta^2)^k  \binom{2k}{k}^{\frac{3}{2}} \nonumber \\
      &\le \beta^{2k} + (8 \alpha \beta^2)^k ,\nonumber \\
      &\le 1,
    \end{align}
    where the first inequality is because $\binom{n-2k}{k} \le \binom{n}{k}^{\frac{1}{2}} \binom{n-k}{k}^{\frac{1}{2}} = \binom{n}{2k}^{\frac{1}{2}} \binom{2k}{k}^\frac{1}{2}$, using the identity $\binom{n}{m}\binom{n-m}{k-m} = \binom{n}{k}\binom{k}{m}$.  The last inequality is for small enough $\alpha$.

    Next we bound above $\E[(p_\beta(D_\alpha)-1) \Id(p_\beta(D_\alpha) > 1)]$.  We in fact bound the greater quantity $\E[p_\beta(D_\alpha) \Id(p_\beta(D_\alpha) > 1)]$.
    Using Cauchy--Schwarz and \cref{eq:OZ-2}, the latter is at most
    \begin{align} \label{eq:OZ-3}
      \E\bigl[p_\beta(D_\alpha)^2\bigr]^{\frac{1}{2}} \Pr\bigl[p_\beta(D_\alpha) > 1\bigr]^{\frac{1}{2}}
      \le 1 \cdot \Pr\bigl[p_\beta(D_\alpha) > 1\bigr]^{\frac{1}{2}} .
    \end{align}
    We will show that 
    \begin{align} \label{eq:OZ-tail}
      \Pr\bigl[p_\beta(D_\alpha) > 1\bigr]^{\frac{1}{2}}
      \le e^{-(\frac{1}{e\beta} - 1)\frac{k}{4}} .
    \end{align}
    So, using $\beta = \frac{100}{\log(1/\alpha)}$, \cref{eq:OZ-3} is less than $(\alpha \beta)^k/2$.
    Plugging these bounds into \cref{eq:OZ-goal}, we conclude that
    \[
      \E\bigl[f_\beta(D_\alpha)\bigr] - \E\bigl[f_\beta(D)\bigr]
      \ge (\alpha \beta)^k/2
    \]
    for any $(2k)$-wise uniform $D$, proving the claim assuming \cref{eq:OZ-tail} holds.

    \medskip

    It remains to prove \cref{eq:OZ-tail}.
    Suppose $\abs{p_\beta(x)} > 1$.
    Then by its definition and \Cref{cor:bijlsv-ineq}, it must be the case that
    \[
      1 
      \le \beta^k \binom{n}{2k}^{-\frac{1}{2}} \abs[\bigg]{ \sum_{\abs{S}=2k} x^S }
      \le \beta^k \binom{n}{2k}^\frac{1}{2} \left(\frac{2k}{n} + \frac{(\summ x)^2}{n^2}\right)^k 
      \le (e \beta)^k \left(1 + \frac{(\summ x)^2}{2kn} \right)^k ,
    \]
    which implies  $x \in E_\beta := \{ x \in \pmo^n : (\summ x)^2 \ge (\frac{1}{e \beta} - 1) 2kn \}$.
    Jumping ahead, we will use below that by Hoeffding's inequality, we have $\Pr[U \in E_\beta] \le e^{-k (\frac{1}{e\beta}-1)}$.
    In the meanwhile, we use the implication just noted to write    
    \begin{align}
      \Pr[p_\beta(D_\alpha) > 1]
      \le \Pr[D_\alpha \in E_\beta]
      &= 2^{-n} \sum_{x \in E_\beta} \Bigl( 1 + \alpha^k \binom{n}{2k}^{-\frac{1}{2}} \sum_{\abs{S}=2k} x^S \Bigr) \nonumber \\
      &= \Pr[ U \in E_\beta] + \alpha^k \binom{n}{2k}^{-\frac{1}{2}} 2^{-n} \sum_{x \in E_\beta} \sum_{\abs{S}=2k} x^S . \label{eq:OZ-cont}
    \end{align}
    We rewrite and bound the second term using Cauchy--Schwarz as follows:
    \begin{align*}
      \alpha^k \binom{n}{2k}^{-\frac{1}{2}} \E\Bigl[ \tsum_{\abs{S}=2k} U^S \cdot \Id(U \in E_\beta)\Bigr]
      &\le
      \alpha^k \binom{n}{2k}^{-\frac{1}{2}} \E\Bigl[ \bigl( \tsum_{ \abs{S}=2k} U^S \bigr)^2 \Bigr]^{\frac{1}{2}}  \cdot \Pr[ U \in E_\beta ]^{\frac{1}{2}} \\
      &= \alpha^k \cdot \Pr[U \in E_\beta]^{\frac{1}{2}} .
    \end{align*}
    Therefore
    \[
      \cref{eq:OZ-cont}
      \le \Pr[ U \in E_\beta]^{\frac{1}{2}} \Bigl(\Pr[U \in E_\beta]^{\frac{1}{2}} + \alpha^k \Bigr) \\
      \le e^{-(\frac{1}{e\beta} - 1)\frac{k}{2}} \cdot (1/2 + 1/2) .
    \]
    This proves \cref{eq:OZ-tail}.
  \end{proof}

  \begin{proof}[Proof of \Cref{thm:sbpn-vs-k-wise}]
    For any $\rho \in (0,1]$, by \Cref{claim:noise-invariant} we have  $D_{c} \cdot N_\rho = D_{c\rho^2}$.
    So we can take $\alpha$ to be $c\rho^2$, and thus $\beta = c/\log(1/2\rho)$.
    By \Cref{claim:OZ-distinguish}, the distinguishing advantage is at least $(c\rho^2/\log(1/2\rho))^k$.
  \end{proof}

\section{Bounded uniformity plus noise fools symmetric tests} \label{sec:bi+noise-fools-sym}

Here we prove \Cref{thm:bipnvssym}.  The starting observation for the proof of this theorem (and also of \Cref{thm:sym-xor-sym-fools-any,thm:improve-Bazzi}) is that the Fourier expansion of any symmetric function is a linear combination of the Krawtchouk polynomials $\sK(\ell,\summ x) := \sum_{\abs{S}=\ell} x^S$ weighted by the coefficients $\hf([\ell])$.
As $k$-wise uniformity fools all parities of size at most $k$, it suffices to consider the $\ell > k$ terms.
While $\sK(\ell,\summ x)$ can be as large as $\binom{n}{\ell}$ on the all-1 string, it follows from Cauchy--Schwarz that its average $\E_x[\abs{\sK(\ell,\summ x)}]$ is at most $\binom{n}{\ell}^{1/2}$.
Moreover, a simple argument (\Cref{fact:coeff-sym-fcn}) shows that $\abs{\hf([\ell])}$ is bounded by $\binom{n}{\ell}^{-1/2}$, the reciprocal of the upper bound on $\E_x[\abs{\sK(\ell, \summ x)}]$, and so their product is at most $1$, which is then dampened to $\rho^\ell \le \rho^k$ by noise.

To make this argument go through, we use \Cref{cor:bijlsv-ineq} to show that $\abs{\sK(\ell,\summ x)}$ is close to $\binom{n}{\ell}^{1/2}$ when $x$ is nearly-balanced, which holds with high probability under $k$-wise uniformity (\Cref{cor:tail-bound}).

We start by proving a few useful facts about symmetric functions and distributions.

\begin{fact} \label{fact:coeff-sym-fcn}
  Let $f\colon\pmo^n \to [-1,1]$ be any symmetric function.
  For every $S \subseteq [n]$ of size $\ell$, we have (1) $\hf(S) = \hf([\ell])$ and (2) $\abs{\hf([\ell])} \le \binom{n}{\ell}^{-1/2}$.
\end{fact}
\begin{proof}
  (1) is clear.
  To see (2), by Cauchy--Schwarz and Parseval's identity, we have
  \[
    \binom{n}{\ell} \abs[\big] {\hf([\ell])}
    = \abs[\Big]{ \sum_{\abs{S}=\ell} \hf(S) }
    \le \binom{n}{\ell}^{1/2} \biggl( \sum_{\abs{S}=\ell} \hf(S)^2 \biggr)^{1/2}
    \le \binom{n}{\ell}^{1/2} \E\bigl[f(U)^2\bigr]
    \le \binom{n}{\ell}^{1/2} . \qedhere
  \]
\end{proof}
We also need the following well-known moment bounds for $k$-wise uniform distributions.
For a short proof see \cite[Lemma 32]{BHLV}.

\begin{lemma} \label{lemma:moment}
  Let $D$ be a $(2k)$-wise uniform distribution on $\pmo^n$.
  Then $\E[ (\sum_{i=1}^n D_i)^{2k} ] \le \sqrt{2} \left(2kn/e\right)^k$.
\end{lemma}

By Markov's inequality, this implies the following tail bound.

\begin{corollary} \label{cor:tail-bound}
  Let $D$ be a $(2k)$-wise uniform distribution on $\pmo^n$.
  For every integer $t > 0$, we have
  \[
  \Pr\bigl[ \abs{\summ D} \ge t \bigr]
    \le \sqrt{2} \left(\frac{2kn}{e t^2}\right)^k .
  \]
\end{corollary}

The following fact says that a distribution remains close to itself after conditioning on any high probability event.

\begin{fact} \label{fact:tvd-conditioning}
  Let $D$ be any distribution on $\pmo^n$ and $E$ be any event.
  Then the conditional distribution $D \mid E$ is $(1 - \Pr[E])$-close to $D$.
\end{fact}

\begin{proof}
  Let $\overline{E}$ be the complement of $E$.
  For every Boolean test $g\colon\pmo^n \to \zo$ we have
  \begin{align*}
    \E[g(D)]
    &= \E[g(D \mid E)] (1 - \Pr[\overline{E}]) + \E[g(D \mid \overline{E})] \Pr[\overline{E}] \\
    &= \E[g(D\mid E)] + \bigl( \E[g(D \mid \overline{E})] - \E[g(D \mid E)] \bigr) \Pr[\overline{E}] .
  \end{align*}
So $\abs{\E[g(D)] - \E[g(D \mid E)]} \le \Pr[\overline{E}]$, as $\abs{\E[g(D \mid \overline{E})] - \E[g(D \mid E)]}$ is bounded by $1$.
\end{proof}

\begin{proof}[Proof of \Cref{thm:bipnvssym}]
  Define $G := \{ x \in \pmo^n: \abs{\sum_{i=1}^n x_i} \le \sqrt{nk/3\rho} \}$.
  We write $f := f_{\le k} + f_{>k}$, where $f_{\le k}(x) = \sum_{\abs{S}\le k} \hf(S) x^S$, and $f_{>k}(x) := f(x) - f_{\le k}(x) = \sum_{\abs{S} > k} \hf(S) x^S$.
  For convenience let $Z:= D \cdot N_\rho$.
  As $Z$ is $(2k)$-wise uniform, we have
  \begin{align*}
      \E [f] = \E\bigl[f_{\le k}(Z) \bigr]  = \E\bigl[f_{\le k}(Z) \Id (D\in G) \bigr] +  \E\bigl[f_{\le k}(Z) \Id (D\notin G) \bigr].
  \end{align*}
  So we can bound the error by
  \begin{align}
    \abs[\big]{ \E[f(Z)] - \E[f] }
    &=  \abs[\big]{ \E\bigl[f(Z) \Id(D \in G) \bigr] +\E\bigl[f(Z) \Id( D\notin G)\bigr] - \E[f]  } \nonumber  \\
    &\le \abs[\big]{ \E\bigl[f_{\le k}(Z) \Id(D \in G)\bigr] - \E[f]  } +\abs[\big]{\E\bigl[f_{>k}(Z) \Id(D \in G)\bigr]} + \Pr [D \notin G] \nonumber \\
    &\le \abs[\big]{ \E\bigl[f_{\le k}(Z) \Id(D \notin G) \bigr]  }  + \abs[\big]{\E\bigl[f_{>k}(Z) \Id(D \in G)\bigr]} + \Pr[D \notin G] \label{eq:bipnfsym-main} ,
  \end{align}

  We now bound each term individually.
  By \Cref{cor:tail-bound}, we have
  \begin{align} \label{eq:bipnfsym-tail}
    \Pr[ D \not\in G ]
    \le \sqrt{2} \cdot \left( \frac{2 \cdot 3\rho}{e} \right)^k 
    \le \sqrt{2} \cdot (e\rho)^k .
  \end{align}

  We now bound the first term, 
  As $f_{\le k}^2$ has degree $2k$, by Parseval's identity and $(2k)$-wise uniformity of $Z$, we have
  \[
    \E[f_{\le k}(Z)^2]
    = \E[f_{\le k}(U)^2]
    = \E[f(U)^2]
    \le 1 .
  \]
  By Cauchy--Schwarz, the first term in \cref{eq:bipnfsym-main} is at most
  \begin{align} \label{eq:bipnfsym-1}
    \abs[\big]{ \E[f_{\le k}(Z) \Id (D\notin G) ]  }
    \le \E[f_{\leq k}(Z)^2 ]^{1/2} \Pr[D \notin G]^{1/2}
    \le 2^{1/4} \cdot (e \rho)^{k/2} .
  \end{align}

  It remains to bound the second term in \cref{eq:bipnfsym-main}.
  For every $x \in G$,  we will show that
  \begin{align} \label{eq:bipnfsym-2}
    \abs[\big]{\E\bigl[f_{>k}(x \cdot N_\rho)\bigr]}
    \le 7 \cdot (e\rho)^{k/2} .
  \end{align}
  Plugging \cref{eq:bipnfsym-tail,eq:bipnfsym-1,eq:bipnfsym-2} into \cref{eq:bipnfsym-main} gives an error bound of at most $10 (e\rho)^{k/2}$, as desired.

  We now show \cref{eq:bipnfsym-2}.
  As $\E[N_\rho^S] = \rho^{\abs{S}}$, we have
  \begin{align*}
    \abs[\Big]{\E[f_{>k}(x \cdot N_\rho)]}
    = \abs[\Big]{ \sum_{\abs{S} > k} \rho^{\abs{S}} \hf(S) x^S }
    &=  \abs[\Big]{ \sum_{\ell=k+1}^n \rho^\ell \sum_{\abs{S}=\ell} \hf(S) x^S } .
  \end{align*}
  Applying \Cref{fact:coeff-sym-fcn,cor:bijlsv-ineq}, and using the inequality $\binom{n}{\ell} \le (en/\ell)^\ell$, we have
  \begin{align*}
    \abs[\Big]{ \sum_{\ell=k+1}^n \rho^\ell \sum_{\abs{S}=\ell} \hf(S) x^S }
    &\le \abs[\Big]{ \sum_{\ell=k+1}^n \rho^\ell \cdot \hf([\ell]) \sum_{\abs{S}=\ell} x^S } \\
    &\le \sum_{\ell=k+1}^n \rho^\ell \cdot \abs[\big]{\hf([\ell])} \cdot  \abs[\Big]{ \sum_{\abs{S}=\ell} x^S } \\
    &\le \sum_{\ell=k+1}^n \rho^\ell \cdot \binom{n}{\ell}^{1/2}  \left( \frac{\ell}{n} + \frac{k}{3\rho n} \right)^{\ell/2} \\
    &\le \sum_{\ell=k+1}^n \rho^\ell \cdot e^{\ell/2} \left(  1 + \frac{k}{3\rho \ell} \right)^{\ell/2} \\
    &\le \sum_{\ell=k+1}^n \left( \rho \left(  \rho e + \frac{e}{3} \right) \right)^{\ell/2} \\
    &\le 7 \cdot (e\rho)^{k/2}
  \end{align*}
  where the last inequality is because we can assume $\rho \le 1/e$, as otherwise the conclusion is trivial, and so we have $\rho e + e/3 \le 1 + e/3 \le 2$.
  This shows \cref{eq:bipnfsym-2}.
\end{proof}

\section{Shifted symmetric small-bias fools symmetric tests} \label{sec:shifted-sym-small-bias}

In this section we prove \Cref{thm:sym-xor-sym-fools-any}.  The proof follows a similar high-level idea to the proof of \Cref{thm:bipnvssym}, but we trade symmetry for noise, because xor-ing the uniform permutation of a string has a similar effect to adding noise (see \Cref{claim:bias-unif-weight}).

As mentioned in the introduction, we actually prove stronger results about fooling symmetric functions.
One can then obtain \Cref{thm:sym-xor-sym-fools-any} by combining \Cref{cor:shifted-sym-bias-special} below and the following claim.

\begin{claim} \label{claim:sym-close-to-bin-implies-non-sym-close}
  Let $D$ be a symmetric distribution on $\zo^n$.
  If $\abs{D}$ is $\eps$-close to the binomial distribution $\Bin(n,1/2)$, then $D$ is $\eps$-close to the uniform distribution.
\end{claim}
\begin{proof}
  We have
  \[
    \sum_{w=0}^n \sum_{\abs{x}=w} \abs*{2^{-n} - \frac{D(w)}{\binom{n}{w}}}
  = \sum_{w=0}^n \binom{n}{w} \abs[\bigg] { 2^{-n} - \frac{D(w)}{\binom{n}{w}} }
  = \sum_{w=0}^n \abs[\bigg] { 2^{-n}\binom{n}{w} - D(w)}
    \le \eps . \qedhere
  \]
\end{proof}

In turn, \Cref{cor:shifted-sym-bias-special} follows from \Cref{thm:shifted-sym-bias-special}, showing that any symmetric small-bias distribution xor-ed a nearly-balanced string fools symmetric functions.
Then we prove \Cref{thm:shifted-sym-bias}, which is a generalization of \Cref{thm:shifted-sym-bias-special} that covers a more general settings of parameters.
We complement \Cref{thm:shifted-sym-bias} with a lower bound (\Cref{claim:shifted-sym-bias-lb}).

First we show that the bias of the uniform permutation of a string on parity tests is equal to the normalized Krawtchouk polynomials.

\begin{claim} \label{claim:bias-unif-weight}
  Let $W_t$ be the uniform distribution on $\{x \in \pmo^n: \sum_{i=1}^n x_i = t \}$.
  For every subset $S \subseteq [n]$ of size $\ell$, we have
  \[
    \abs[\big]{ \E\bigl[ W_t^{[n]\setminus S} \bigr] }
    = \abs[\big]{ \E\bigl[ W_t^S \bigr] }
    = \frac{\abs[\big]{\sK(\ell,t)}}{\binom{n}{\ell}} .
  \]
\end{claim}
\begin{proof}
  The first equality follows from $\abs{\sum_{\abs{S}=n-\ell} z^S} = \abs{ z^{[n]} \sum_{\abs{S}=\ell} z^S } = \abs{\sum_{\abs{S}=\ell} z^S}$.
  To prove the second inequality, first fix a string $z$ with $\sum_{i=1}^n z_i = t$.
  Observe that by symmetry we have $\sum_{x: \sum_{i=1}^n x_i = t} x^S = \sum_{x: \sum_{i=1}^n x_i = t} x^{[\ell]}$ for any $S \subseteq [n]$ of size $\ell$, and $\sum_{\abs{S}=\ell} x^S = \sum_{\abs{S}=\ell} z^S$ for any $x \in \pmo^n$ with $\sum_{i=1}^n x_i = t$.
  Hence,
  \[
    \binom{n}{\ell} \sum_{x: \sum_i x_i = t} x^{[\ell]}
    = \sum_{\abs{S} = \ell} \sum_{x: \sum_i x_i = t} x^S
    = \sum_{x: \sum_i x_i = t} \sum_{\abs{S} = \ell} x^S
    = \binom{n}{t} \sum_{\abs{S}=\ell} z^S .
  \]
  Rearranging gives
  \[
    \abs[\big]{\E[W_t^S]}
    = \frac{1}{\binom{n}{t}} \abs[\Big]{ \sum_{x: \sum_i x_i = t} x^{[\ell]}} 
    = \frac{1}{\binom{n}{\ell}} \abs[\Big]{\sum_{\abs{S}=\ell} z^{S}}
    = \frac{\abs[\big]{\sK(\ell,t)}}{\binom{n}{\ell}} . \qedhere
  \]
\end{proof}

\subsection{Proof of \Cref{cor:shifted-sym-bias-special}}

\Cref{cor:shifted-sym-bias-special} is a straightforward corollary of \Cref{thm:shifted-sym-bias-special}, which we now prove.

\begin{theorem} \label{thm:shifted-sym-bias-special}
  Let $D_\sym$ be a symmetric $n^{-20k}$-biased distribution on $\pmo^n$ and $z \in \pmo^n$ be any string with $\abs{\sum_{i=1}^n z_i} \le n^{0.6}$.
  Then $\abs{\E[f(z \cdot D_\sym)] - \E[f]} \le O_k(n^{-0.3k})$.
\end{theorem}

Note that it is crucial that $D_\sym$ is symmetric, as small-bias distributions are closed under shifts; so every small-bias distribution $D$ is also a shifted small-bias distribution.

\begin{corollary} \label{cor:shifted-sym-bias-special}
Let $D_\sym$ and $D$ be two independent $n^{-20k}$-biased distributions on $\pmo^n$, where $D_\sym$ is symmetric.
Then $\abs{\E[f(D_\sym + D)] - \E[f]} \le c_k(n^{-0.3k})$ for every symmetric function $f\colon\pmo^n \to [-1,1]$.
\end{corollary}

Also note that $D + D_\sym$ itself is not necessarily a symmetric distribution.

\begin{proof}[Proof of \Cref{cor:shifted-sym-bias-special}]
  By \Cref{lemma:OZ}, $D$ is $n^{-10k}$-close to $(10k)$-wise uniform.
  So by \Cref{cor:tail-bound},
  \[
    \Pr\biggl[ \abs[\Big]{\sum_{i=1}^n D_i} \ge n^{0.6} \biggr] 
    \le \left(\frac{10kn}{n^{1.2}}\right)^{5k} + n^{-10k}
    \le O_k(n^{-k}) .
  \]
  It follows from \Cref{thm:shifted-sym-bias-special} that the error is $O_k(n^{-k}) + O_k(n^{-0.3k}) = O_k(n^{-0.3k})$.
\end{proof}

\begin{proof}[Proof of \Cref{thm:shifted-sym-bias-special}]
  Let $\eps := n^{-20k}$ be the bias of $D_\sym$ and $t := n^{0.6}$.
  Define $G := \{ x \in \pmo^n: \abs{\sum_{i=1}^n x_i} \le t \}$.
  As $D_\sym$ is $\eps$-biased, by \Cref{lemma:OZ}, it is $\delta$-close to $(30k)$-wise uniform, where $\delta \le n^{-4k}$.
  Applying \Cref{cor:tail-bound} with our choice of $t = n^{0.6}$ in the definition of $G$, we have
  \begin{align}  \label{eq:shifted-sym-bias-sp-tail}
    \Pr[ D_\sym \not\in G ]
    \le \left( \frac{ 30kn }{n^{1.2}} \right)^{15k} + \delta
    \le O_k(n^{-3k}) .
  \end{align}
  Let $D_\sym'$ be the distribution of $D_\sym$ conditioned on $D_\sym \in G$.
  Note that $D_\sym'$ remains a symmetric distribution and by \Cref{fact:tvd-conditioning} is $\Pr[D_\sym \notin G]$-close to $D_\sym$, and thus is $\eps'$-biased, where $\eps' := \eps + \Pr[D_\sym \notin G] \le O_k(n^{-3k})$.
  We now write $f := f_{\middle} + f_{\ends}$, where $f_{\middle}(x) := \sum_{k <  \abs{S} < n-k} \hf(S) x^S$, and $f_\ends(x) := f(x) - f_{\middle}(x) = \sum_{\abs{S} \in [0,k] \cup [n-k,n]} \hf(S) x^S$.
  By the triangle inequality, we have 
  \begin{align}
    \abs[\big]{\E[f(z \cdot D_\sym)] - \E[f]}
    &\le \abs[\big]{\E[f(z \cdot D_\sym')] - \E[f]} + \Pr[D_\sym \notin G] \nonumber \\
    &\le \abs[\big]{\E[f_\ends(z \cdot D_\sym')] - \E[f]} + \abs[\big]{\E[f_\middle(z \cdot D_\sym')]} + O_k(n^{-3k}) . \label{eq:shifted-sym-bias-sp-cases}
  \end{align}
  We now bound each of the two terms on the right hand side.
  As $z \cdot D_\sym'$ is $\eps'$-biased, we have
  \begin{align} \label{eq:shifted-sym-bias-sp-ends}
    \abs[\big]{ \E[f_{\ends}(z \cdot D_\sym')] - \E[f] }
    \le \sum_{\abs{S} \in [1,k] \cup [n-k,n]} \abs{\hf(S)} \eps'
    \le 2 n^k \eps'
    \le O_k(n^{-2k}) .
  \end{align}
  
  To bound $\abs{\E[f_\middle(z \cdot D_\sym')]}$, let $S$ be any subset of size $\ell$.
  As $f$ is symmetric, we have $\hf(S) = \hf([\ell])$.
  As $D_\sym'$ is also symmetric, we have $\E[D_\sym'^S] = \eps_\ell$ for some $\eps_\ell$ which only depends on the size of $S$.
  Hence,
  \begin{align*}
  \abs[\Big]{\E[f_{\middle}(z \cdot D_\sym')}
  \le  \abs[\bigg]{ \sum_{\ell=k+1}^{n-k-1} \sum_{\abs{S}=\ell} \hf(S) \E[D_\sym'^S] z^S }
  =  \abs[\bigg]{\sum_{\ell=k+1}^{n-k-1} \hf([\ell]) \eps_\ell \sum_{\abs{S}=\ell} z^S } .
  \end{align*}
  As $\abs{\sum_{\abs{S}=\ell} z^S} = \abs{ z^{[n]} \sum_{\abs{S}=n-\ell} z^S} = \abs{\sum_{\abs{S}=n-\ell} z^S}$, by \Cref{fact:coeff-sym-fcn,claim:bias-unif-weight} we have
  \[
    \abs[\bigg]{ \sum_{\ell=k+1}^{n-k-1} \hf([\ell]) \eps_\ell \sum_{\abs{S}=\ell} z^S }
  \le \sum_{\ell=k+1}^{n-k+1} \biggl( \abs{\hf([\ell])} \cdot \abs{\eps_\ell} \cdot \abs[\bigg]{\sum_{\abs{S}=\ell} z^S} \biggr) \\
  \le 2 \sum_{\ell=k+1}^{\floor{n/2}} \frac{\sK(\ell,n^{0.6})^2}{\binom{n}{\ell}^{3/2}} .
  \]
  We first bound the partial sum over $\ell$ from $n/4$ to $n/2$.
  Note that the binary entropy function $H(x) := x \log_2(1/x) + (1-x) \log_2(1/(1-x))$ is increasing on $[0,1/2]$.
  In particular, we have $H(1/4) \ge 4/5$ and so $\frac{3}{2}H(1/4) \ge 6/5$.
  By Stirling's approximation, we have $\binom{n}{\ell} \ge \frac{1}{n^2} 2^{n H(\frac{\ell}{n})}$ (see \cite{CT06} for a proof).
  Applying \Cref{cor:Harm-bound} with these facts, we have
  \begin{align}
    \sum_{\ell=\max\{n/4,k\}+1}^{\floor{n/2}} \frac{\sK(\ell,n^{0.6})^2}{\binom{n}{\ell}^{3/2}}
    \le \frac{n}{4} \cdot n^2 \cdot 2^{-n \bigl( \frac{3}{2} H(\frac{1}{4})- 1 - n^{-0.8} \bigr) }
    \le 2^{-n/10} .
  \end{align}  
  We now bound the remaining sum (i.e., the partial sum from $\ell=k+1$ to $n/4$).
  Using \Cref{cor:bijlsv-ineq} and the inequality $\binom{n}{\ell} \le (\frac{en}{\ell})^\ell$, we have
  \begin{align*}
    \sum_{\ell=k+1}^{n/4} \frac{\sK(\ell,n^{0.6})^2}{\binom{n}{\ell}^{3/2}}
    \le \sum_{\ell=k+1}^{n/4} \binom{n}{\ell}^{1/2} \left( \frac{\ell}{n} + \frac{n^{1.2}}{n^2} \right)^{\ell}
    \le \sum_{\ell=k+1}^{n/4} \left(\frac{en}{\ell}\right)^{\ell/2} \left( \frac{\ell}{n} + \frac{1}{n^{0.8}} \right)^{\ell} .
  \end{align*}
  Observe that each term in the sum is at most $1/2$ of its previous term, and so this sum is bounded by twice the first term, which is at most $O_k(n^{-0.3k})$.
  Therefore, 
  \begin{align} \label{eq:shifted-sym-bias-sp-middle}
    \abs[\big]{ \E\bigl[f_\middle(z \cdot D_\sym')\bigr]}
    \le 2^{-n/10} + O_k(n^{-0.3k})
    \le O_k(n^{-0.3k}) .
  \end{align}
  Plugging \cref{eq:shifted-sym-bias-sp-ends,eq:shifted-sym-bias-sp-middle} in \cref{eq:shifted-sym-bias-sp-cases} completes the proof.
  \end{proof}

\subsection{General case}

\Cref{thm:shifted-sym-bias-special} is stated for a nearly-balanced shift.
We now prove a general bound that holds for any shifts.

\begin{theorem} \label{thm:shifted-sym-bias}
  There exists a constant $C$ such that the following holds.
  Let $D_\sym$ be a symmetric $\eps$-biased distribution on $\pmo^n$ and $z \in \pmo^n$ be any string.
  Let $s := \abs{\sum_{i=1}^n z_i}$.
  For every positive integer $k$ and every symmetric function $f\colon\pmo^n \to [-1,1]$, we have
  \[
    \abs[\bigg]{ \E[f(z \cdot D_\sym)] - \E[f] }
    \le C \left( \left( \frac{11\max\{s,\sqrt{kn}\}}{n} \right)^{k/2}  + \left(\frac{e^3n}{2k}\right)^{k/2} \eps \right) .
  \]
\end{theorem}

The following lower bound shows that the dependence on $s$ in \Cref{thm:shifted-sym-bias} is necessary.

\begin{claim} \label{claim:shifted-sym-bias-lb}
  There exists a constant $c > 0$ such that the following holds.
  For every integer $m \ge 3$, there is a symmetric $e^{-cn/m^2}$-biased distribution $D$ on $\zo^n$ such that for every string $z \in \zo^n$ of Hamming weight at most $\floor{m/2}-1$, there exists a symmetric function $f\colon \zo^n \to \zo$ such that $f(D)=0$ always and $\Pr[f(U) = 1] \ge 1/m - e^{-cn/m^2}$.
\end{claim}


\begin{proof}[Proof of \Cref{claim:shifted-sym-bias-lb}]
  Let $D$ be the uniform distribution on $\{ x \in \zo^n: \sum_i x_i \equiv 0 \bmod m\}$.
  It is known that $D$ is $2^{-cn/m^2}$-biased (see Claim~18 and Lemma 19 in \cite{BIJLSV21} for a proof). Let $z$ be any string of weight $|z| \le \floor{m/2}-1$.
  Consider the symmetric function $f(x) := \Id(|x| \equiv \floor{(m+1)/2} \bmod m)$.
  By the triangle inequality, we have $|D|-|z| \le |D+z| \le |D|+|z|$, and so $|D+z| \not\equiv \floor{(m+1)/2} \pmod m$.

  On the other hand, it is known that $\Pr[\Bin(n,1/2) \equiv \floor{(m+1)/2}] \ge 1/m - e^{-cn/m^2}$ (see, again, Claim~18 in \cite{BIJLSV21} for a proof).
\end{proof}

\begin{proof}[Proof of \Cref{thm:shifted-sym-bias}]
  We may assume $k \le n/16$ and $s \le n/120$, as otherwise the bound given in the theorem is at least $1$.
  Define $G := \{ x \in \pmo^n: \abs{\sum_{i=1}^n x_i} \le t \}$, where $t := t(n,k,s)$ is a parameter to be chosen.
  As $D_\sym$ is $\eps$-biased, by \Cref{lemma:OZ}, it is $\delta$-close to $(2k)$-wise uniform, where $\delta := (\frac{e^3n}{2k})^k \eps$.
  Applying \Cref{cor:tail-bound}, we have
  \begin{align} \label{eq:shifted-sym-bias-tail}
    \Pr[ D_\sym \not\in G ]
    \le \sqrt{2} \cdot \left( \frac{2nk}{et^2} \right)^k + \delta .
  \end{align}
  We write $f := f_{\middle} + f_{\ends}$, where $f_{\middle}(x) := \sum_{k+1 < \abs{S} < n-k} \hf(S) x^S$, and $f_\ends(x) := f(x) - f_{\middle}(x) = \sum_{\abs{S} \in [0,k] \cup [n-k,n]} \hf(S) x^S$.
  For convenience, let $Z:= z \cdot D_{\sym}$.
  As $Z$ is $\eps$-biased, we have
  \[
    \abs[\big]{ \E[f_{\ends}(Z)] - \E[f] }
    \le \sum_{\abs{S} \in [1,k] \cup [n-k,n]} \abs{\hf(S)} \eps
    \le 2 \biggl(\frac{e^3n}{k}\biggr)^{k/2} \eps
    \le \delta .
  \]
  By the triangle inequality, we have
  \begin{align}
    \MoveEqLeft
    \abs[\big]{ \E\bigl[f_{\ends}(Z) \Id(D_\sym \in G)\bigr] - \E[f]  }  \nonumber \\
    &\le \abs[\big]{ \E\bigl[f_{\ends}(Z) \Id(D_\sym \in G)\bigr] - \E\bigl[f_{\ends}(Z)\bigr]  } + \abs[\big]{ \E[f_{\ends}(Z)] - \E[f] } \nonumber \\
    &\le \abs[\big]{ \E\bigl[f_{\ends}(Z) \Id(D_\sym \in G)\bigr] - \E[f_{\ends}(Z)]  } + \delta \nonumber \\
    &= \abs[\big]{ \E\bigl[f_{\ends}(Z) \Id(D_\sym \notin G)\bigr] } + \delta . \label{eq:shifted-sym-bias-1}
  \end{align}
  As $Z$ is $\eps$-biased,
  \begin{align*}
    \E\bigl[f_{\ends}(Z)^2\bigr]
    &= \sum_{\abs{S},\abs{T} \in [0,k] \cup [n-k,n]} \hf(S) \hf(T) \E\bigl[Z^{S \symdiff T}\bigr] \\
    &\le \sum_{\abs{S} \in [0,k] \cup [n-k,n]} \hf(S)^2 + \sum_{\abs{S} \ne \abs{T} \in [0,k] \cup [n-k,n]} \abs{\hf(S)} \abs{\hf(T)} \eps  \\
    &\le 1 + 2\binom{n}{k} \eps \le 1 + \delta .
  \end{align*}
  By Cauchy--Schwarz,
  \begin{align} \label{eq:shifted-sym-bias-2}
    \abs[\big]{ \E\bigl[f_{\ends}(Z) \Id (D_\sym \notin G) \bigr]  }
    \le \E[f_{\ends}(Z)^2 ]^{1/2} \Pr[D_\sym \notin G]^{1/2}
    \le 2 \Pr[D_\sym \notin G]^{1/2} .
  \end{align}
  We now use \cref{eq:shifted-sym-bias-1,eq:shifted-sym-bias-2} to bound the error as follows:
  \begin{align}
    \abs[\big]{ \E[f(Z)] - \E[f] }
    &=  \abs[\big]{ \E\bigl[f(Z) \Id(D_\sym \in G) \bigr] +\E\bigl[f(Z) \Id( D_\sym \notin G)\bigr] - \E[f]  } \nonumber  \\
    &\le \abs[\big]{ \E\bigl[f_{\ends}(Z) \Id(D_\sym \in G)\bigr] - \E[f]  } +\abs[\big]{\E\bigl[f_{\middle}(Z) \Id(D_\sym \in G)\bigr]} + \Pr [D_\sym \notin G] \nonumber  \\
    &\le \abs[\big]{ \E\bigl[f_{\ends}(Z) \Id(D_\sym \notin G)\bigr]  }  + \abs[\big]{\E\bigl[f_{\middle}(Z) \Id(D_\sym \in G)\bigr]} + \Pr[D_\sym \not\in G] + \delta \nonumber \\ 
    &\le \abs[\big]{ \E\bigl[f_{\middle}(Z) \Id(D_\sym \in G)\bigr] } + 3 \Pr[D_\sym \not\in G]^{1/2} + 2\delta . \label{eq:shifted-sym-bias-main}
  \end{align}
  We will bound the first term in \cref{eq:shifted-sym-bias-main} by
  \begin{align} \label{eq:shifted-sym-bias-mid}
    \abs[\big]{\E[f_{\middle}(z \cdot D_{\sym})]}
    \le \begin{cases} O(1) \left( \frac{120k}{n} \right)^{k/4} & \text{if $s \le \sqrt{kn}$ and $t = (k n^3)^{1/4}$} \\
    O(1) \left(\frac{120s^2}{n^2}\right)^{k/4} & \text{if $s \ge \sqrt{kn}$ and $t = \left(\frac{k^2n^4}{s^2}\right)^{1/4}$.} \end{cases}
  \end{align}
  Plugging \cref{eq:shifted-sym-bias-tail} and \cref{eq:shifted-sym-bias-mid} into \cref{eq:shifted-sym-bias-main} gives us an error of 
  \[
    O(1) \left( \left( \frac{120 \max\{s, \sqrt{kn}\}}{n} \right)^{k/2} + \delta \right) 
  \]
  as desired.

  \medskip

  It remains to prove \cref{eq:shifted-sym-bias-mid}.
  Let $S$ be any subset of size $\ell$.
  As $f$ is symmetric, we have $\hf(S) = \hf([\ell])$.
  Let $D_\sym'$ be the distribution of $D_\sym$ conditioned on $D_\sym \in G$.
  Note that $D_\sym'$ is also symmetric, and so we have $\E[D_\sym'^S] = \eps_\ell$ for some $\eps_\ell$ which only depends on the size of $S$.
  Hence,
  \begin{align*}
  \abs[\Big]{\E[f_{\middle}(z \cdot D_{\sym}) \Id(D_\sym \in G)]}
  \le  \abs[\bigg]{ \sum_{\ell=k+1}^{n-k-1} \sum_{\abs{S}=\ell} \hf(S) \E[D_\sym'^S] z^S }
  =  \abs[\bigg]{\sum_{\ell=k+1}^{n-k-1} \hf([\ell]) \eps_\ell \sum_{\abs{S}=\ell} z^S } .
  \end{align*}
  As $\abs{\sum_{\abs{S}=\ell} z^S} = \abs{ z^{[n]} \sum_{\abs{S}=n-\ell} z^S} = \abs{\sum_{\abs{S}=n-\ell} z^S}$, by \Cref{fact:coeff-sym-fcn,claim:bias-unif-weight} we have
  \[
    \abs[\Big]{ \sum_{\ell=k+1}^{n-k-1} \hf([\ell]) \eps_\ell \sum_{\abs{S}=\ell} z^S }
    \le \sum_{\ell=k+1}^{n-k+1} \biggl( \abs{\hf([\ell])} \cdot \abs{\eps_\ell} \cdot \abs[\bigg]{\sum_{\abs{S}=\ell} z^S} \biggr) \\
    \le 2 \sum_{\ell=k+1}^{\floor{n/2}} \frac{\sK(\ell,t) \sK(\ell, s)    }{\binom{n}{\ell}^{3/2}}  .
  \]
  We first bound the sum over $\ell$ from $n/4$ to $n/2$.
  Note that the binary entropy function $H(x) := x \log_2(1/x) + (1-x) \log_2(1/(1-x))$ is increasing on $[0,1/2]$.
  In particular, we have $H(1/4) \ge 4/5$ and so $\frac{3}{2}H(1/4) \ge 6/5$.
  By Stirling's approximation, we have $\binom{n}{\ell} \ge \frac{1}{n^2} 2^{n H(\frac{\ell}{n})}$ (see \cite{CT06} for a proof).
  Applying \Cref{cor:Harm-bound} with these facts along with $s \le n/120$ and $t \le (kn^3)^{1/4} \le n/2$, we have
  \begin{align} \label{eq:shifted-sym-bias-Harm}
    \sum_{\ell=\max\{n/4,k\}+1}^{\floor{n/2}} \frac{\sK(\ell,t) \sK(\ell, s)    }{\binom{n}{\ell}^{3/2}}
    &\le \sum_{\ell=\max\{n/4,k\}+1}^{\floor{n/2}} \frac{1}{\binom{n}{\ell}^{3/2}} \cdot 2^{\frac{n}{2}(2H(\frac{1}{4})+\frac{s^2}{n^2}+\frac{t^2}{n^2})} \\
    &\le \frac{n}{4} \cdot n^2 \cdot 2^{-n \bigl( \frac{3}{2} H(\frac{1}{4}) - 1 - \frac{s^2+t^2}{2n^2} \bigr) }
    \le 2^{-n/15} .
  \end{align} 
  We now bound the remaining sum (i.e. from $\ell=k+1$ to $n/4$).
  Using \Cref{cor:bijlsv-ineq,claim:bias-unif-weight}, and the inequality $\binom{n}{\ell} \le (\frac{en}{\ell})^\ell$, we have
  \begin{align} 
    \sum_{\ell=k+1}^{\floor{n/2}} \frac{\sK(\ell,t) \sK(\ell, s)    }{\binom{n}{\ell}^{3/2}}
    &\le \sum_{\ell=k+1}^{n/4} \binom{n}{\ell}^{1/2} \cdot \left( \frac{\ell}{n} + \frac{t^2}{n^2} \right)^{\ell/2} \left( \frac{\ell}{n} + \frac{s^2}{n^2} \right)^{\ell/2} \nonumber \\
    &= \sum_{\ell=k+1}^{n/4} \left( e \left( \frac{\ell}{n} + \frac{t^2}{n^2} + \frac{s^2}{n^2} + \frac{t^2s^2}{n^3\ell} \right) \right)^{\ell/2} .\label{eq:shifted-sym-bias-cases}
  \end{align}
  We now consider the two cases in \cref{eq:shifted-sym-bias-mid}.
  \paragraph{Case 1: $s \le \sqrt{nk}$ and $t = (n^3k)^{1/4}$.}
  In this case \cref{eq:shifted-sym-bias-cases} is at most
  \begin{align*}
    \sum_{\ell=k+1}^{n/4} \left( e \left( \frac{\ell}{n} + \sqrt{\frac{k}{n}} + \frac{k}{n} + \sqrt{\frac{k}{n}}  \right) \right)^{\ell/2}
    \le \sum_{\ell=k+1}^{n/4} \left( e \left( \frac{\ell}{n} + 3\sqrt{\frac{k}{n}}\right) \right)^{\ell/2}
    \le O(1) \left( \frac{120k}{n} \right)^{k/4} ,
  \end{align*}
  where the last inequality follows because each term in the sum is at most $9/10$ of its previous term, and so the sum is bounded by $10$ times the first term.
  Combining this with \cref{eq:shifted-sym-bias-Harm} proves the first case in \cref{eq:shifted-sym-bias-mid}.

  \paragraph{Case 2: $s \ge \sqrt{nk}$ and $t = (k^2n^4/s^2)^{1/4}$.}
  In this case \cref{eq:shifted-sym-bias-cases} is at most
   \begin{align*}
    \sum_{\ell=k+1}^{n/4} \left( e \left( \frac{\ell}{n} + \frac{k}{s} + \frac{s^2}{n^2} + \frac{s}{n} \right) \right)^{\ell/2}
    \le \sum_{\ell=k+1}^{n/4} \left( e \left( \frac{\ell}{n} + \frac{3s}{n} \right) \right)^{\ell/2} 
    \le O(1) \left( \frac{120s^2}{n^2} \right)^{k/4} ,
  \end{align*}
  where again the last inequality follows because each term in the sum is at most $9/10$ of its previous term, and so the sum is bounded by $10$ times the first term.
  Combining this with \cref{eq:shifted-sym-bias-Harm} proves the second case in \cref{eq:shifted-sym-bias-mid}.
  \end{proof}

\section{Proof of \Cref{thm:improve-Bazzi}} \label{sec:Bazzi}
In this section, we prove \Cref{thm:improve-Bazzi}, which is based on the same idea that was used in the previous sections.  The difference is that here we use that a typical shift is nearly balanced, and so $\sK(\ell,\summ x)$ is small.

\begin{proof}[Proof of \Cref{thm:improve-Bazzi}]
  Applying Cauchy--Schwarz, Parseval's identity (to the function $g(u) := f(u \cdot D)$), and the assumption that $\E[D^S] = 0$ for $\abs{S} \in [1,k] \cup [n-k,n]$, we have
  \[
    \E_{\bu}\Bigl[ \abs[\big]{ \E[f(\bu \cdot D)] - \E[f] } \Bigr]^2
    \le \E_{\bu}\Bigl[ \bigl( \E[f(\bu \cdot D)] - \E[f] \bigr)^2 \Bigr]
    = \sum_{S: \abs{S} \in (k,n-k)} \hf(S)^2 \E[\chi_S(D^2)] ,
  \]
  where $D^2$ is the sum of two independent copies of $D$, which is also $k$-wise uniform.
  Let $G := \{x \in \pmo^n : \abs{\sum_{i=1}^n x_i} \le (\frac{kn^3}{2e})^{1/4} \}$, and $D_G$ be the conditional distribution of $D^2$ supported on $G$.
  By \Cref{fact:tvd-conditioning}, the distribution $D_G$ is $\Pr[D\notin G]$-close to $D^2$.
  As $\sum_{S \subseteq [n]} \hf(S)^2 \le 1$, we have
  \[
    \abs[\bigg]{ \sum_{S: \abs{S} \in (k,n-k)} \hf(S)^2 \E[(D^2)^S] }
    \le \abs[\bigg]{ \sum_{S: \abs{S} \in (k,n-k)} \hf(S)^2 \E[D_G^S] } + \Pr[D \not\in G] .
  \]
  Applying \Cref{cor:tail-bound} (to the even integer $k-1$ or $k$), we have
  \begin{align} \label{eq:Bazzi-tail}
    \Pr[ D \notin G ]
    \le \left( \frac{2k}{en} \right)^{\frac{k-1}{4}} .
  \end{align}
  We now bound the first term on the right hand side as follows.
  Fix a string $z \in G$.
  As $\abs{\sum_{\abs{S}=\ell} z^S} = \abs{ z^{[n]} \sum_{\abs{S}=n-\ell} z^S} = \abs{\sum_{\abs{S}=n-\ell} z^S}$, by \Cref{fact:coeff-sym-fcn}, 
  \[
    \abs[\bigg]{ \sum_{S: \abs{S} \in (k,n-k)} \hf(S)^2 z^S }
    = \sum_{\ell=k+1}^{n-k-1} \hf([\ell])^2 \abs[\bigg]{\sum_{\abs{S}=\ell} z^S}
    \le 2 \sum_{\ell=k+1}^{n/2} \frac{1}{\binom{n}{\ell}} \abs[\bigg]{\sum_{\abs{S}=\ell} z^S} .
  \]
  We separate the sum into two parts depending on $\ell \le n/5$ and bound each of them individually.
  First, using \Cref{cor:bijlsv-ineq}, we have
  \begin{align} \label{eq:Bazzi-1}
    \sum_{\ell=k+1}^{n/2} \frac{1}{\binom{n}{\ell}} \abs[\bigg]{\sum_{\abs{S}=\ell} z^S}
    \le \sum_{\ell=k+1}^{n/5} \left( \frac{\ell}{n} + \sqrt{\frac{k}{2en}} \right)^{\ell/2}
    \le 2 \left( 2 \sqrt{\frac{k}{2en}} \right)^{k/2} 
    \le 2 \left( \frac{2k}{en} \right)^{k/4} ,
  \end{align}
  because each term in the sum is at most half its previous term, and so the sum can be bounded by twice the first term.
  For the remaining sum (from $\ell=\max\{k,n/5\}+1$ to $n/2$), note that the binary entropy function $H(x) := x \log_2(1/x) + (1-x) \log_2(1/(1-x))$ is increasing on $[0,1/2]$.
  In particular, we have $H(1/5) - \frac{1}{\sqrt{2e}} \ge 1/4$.
  By Stirling's approximation, we have $\binom{n}{\ell} \ge \frac{1}{n^2} 2^{n H(\frac{\ell}{n})}$ (see \cite{CT06} for a proof).
  Applying \Cref{cor:Harm-bound} with these facts, we have
  \begin{align*}
    \sum_{\ell=\max\{n/5,k\}+1}^{n/2} \frac{1}{\binom{n}{\ell}} \abs[\bigg]{\sum_{\abs{S}=\ell} z^S}
    \le \sum_{\ell=\max\{n/5,k\}+1}^{n/2} n^2 \cdot 2^{-\frac{n}{2}(H(\frac{\ell}{n}) - \frac{1}{\sqrt{2e}})}
    \le 2^{-n/10} .
  \end{align*}
  Combining this with \cref{eq:Bazzi-tail,eq:Bazzi-1} gives an error of $(\frac{2k}{en})^{\frac{k-1}{4}} + 2(2(\frac{2k}{en})^{\frac{k}{4}} + 2^{-n/10}) \le 6 (\frac{2k}{en})^{\frac{k-1}{4}}$. \qedhere
\end{proof}

\section{Bounds on Krawtchouk polynomials} \label{sec:Kraw-bounds}

In this section, we prove our upper and lower bounds on Krawtchouk polynomials (\Cref{cor:bijlsv-ineq,prop:Harm-bound-H,claim:sK-lb}).  \Cref{cor:bijlsv-ineq} follows directly from \Cref{lemma:elem-sym-ineq}, which is a general upper bound on the elementary symmetric polynomials $\sum_{\abs{S}=\ell} y^S$ that holds for arbitrary real tuples $y \in \R^n$, not only for $y \in \pmo^n$.

\begin{lemma} \label{lemma:elem-sym-ineq}
  Let $y = (y_1, \ldots, y_n) \in \R^n$.
  For every $1 \le \ell \le n$, we have
  \[
  \abs[\Bigg]{ \sum_{S \subseteq [n]: \abs{S}=\ell} y^S }
    \le \binom{n}{\ell}  \biggl( \frac{\ell-1}{n-1} \cdot \frac{\sum_{i=1}^n y_i^2}{n} + \left(1 - \frac{\ell-1}{n-1} \right) \cdot \frac{\bigl(\sum_{i=1}^n y_i\bigr)^2}{n^2} \biggr)^{\frac{\ell}{2}} .
  \]
  with equality if and only if $y_1 = \cdots = y_n$ or $\ell=1$.
\end{lemma}

Specializing to $y \in \pmo^n$, the elementary symmetric polynomial $\sum_{\abs{S}=\ell} y^S$ is simply the degree-$\ell$ (shifted) Krawtchouk polynomial $\sK(\ell,\abs{y})$.
In this case, we always have $\sum_{i=1}^n y_i^2 = n$, and hence we obtain \Cref{cor:bijlsv-ineq}.

\Cref{cor:bijlsv-ineq} appeared in \cite{BIJLSV21}
with an extra factor of $c^\ell$.  \Cref{lemma:elem-sym-ineq} shows that the same inequality holds even when $y_1, \ldots, y_n$ are arbitrary real numbers.
A similar-looking but incomparable inequality, first proved in \cite{GY20}, showed that
\begin{align} \label{eq:GY-ineq}
  \abs[\Bigg]{ \sum_{S \subseteq [n]: \abs{S}=\ell} y^S }
  \le O\left(\frac{k}{\ell}\right)^{\frac{\ell}{2}} \max_{k'\in\{k,k+1\}} \left(\abs[\Big]{\sum_{S \subseteq [n]: \abs{S}=k'} y^S } \right)^{\frac{\ell}{k'}} ,
\end{align}
Using a different approach, Tao~\cite{Tao23-Maclaurin} recently sharpened this inequality to
\begin{align} \label{eq:Tao-ineq}
  \abs[\Bigg]{ \sum_{S \subseteq [n]: \abs{S}=\ell} y^S }
\le O\left(\frac{1}{\ell}\right)^{\frac{\ell}{2}} \max_{k'\in\{k,k+1\}} \left(\abs[\Big]{\sum_{S \subseteq [n]: \abs{S}=k'} y^S } \right)^{\frac{\ell}{k'}} ,
\end{align}
confirming a conjecture made on MathOverflow, see \url{https://mathoverflow.net/q/446254}.
Note that specializing to the case $k=1$, and using the inequality $\abs{\sum_{1\le i < j \le n} y_i y_j} \le \frac{1}{2}\sum_{i=1}^n y_i^2$, both \Cref{eq:GY-ineq,eq:Tao-ineq} imply a weaker form of \Cref{lemma:elem-sym-ineq}.
In the other direction, Tao~\cite{Tao23-private} observed that one cannot replace the quantity $\sum_{i=1}^n y_i^2$ in \Cref{lemma:elem-sym-ineq} with $\abs{\sum_{1\le i < j \le n} y_i y_j}$, as otherwise when $n$ is the square of an even number, for $y \in \pmo^n$ such that $\sum_{i=1}^n y_i = \sqrt{n}$, we have $\sum_{1 \le i < j \le n} y_i y_j = 0$ and the inequality fails at $\ell = n$.

We note that \Cref{lemma:elem-sym-ineq} can be obtained by a slight modification of both proofs in \cite{GY20,Tao23-Maclaurin}.
Here we follow the approach taken in \cite{Tao23-Maclaurin}, as it gives a sharper constant and the argument is cleaner.

\subsection{Proof of \Cref{lemma:elem-sym-ineq}}
Our approach is based on \cite{Tao23-Maclaurin}, which relies on several basic properties of real-rooted polynomials.
We say an $(n+1)$-tuple of real numbers $(s_0, \ldots, s_n)$ is \emph{attainable} if the polynomial
\[
  \sum_{k=0}^n (-1)^k \binom{n}{k} s_k z^{n-k}
\]
is monic with all roots real.
By its real-rootedness, we can factor the polynomial as
\[
  \sum_{k=0}^n (-1)^k \binom{n}{k} s_k(y) z^{n-k}
  = \prod_{i=1}^n (z-y_i)
\]
for some real numbers $y_1, \ldots, y_n$, where
\[
  s_k(y)
  = \frac{1}{\binom{n}{k}} \sum_{\abs{S}=k} y^S
  = \frac{1}{\binom{n}{k}} \sum_{1 \le i_1 < \cdots < i_k \le n} y_{i_1} \cdots y_{i_k} .
\]
Conversely, given an $n$-tuple of real numbers $y = (y_1, \ldots, y_n)$, we can define $s_k(y)$ as above to obtain an attainable tuple.
We will use the following truncation property of attainable tuples.

\begin{fact}[Truncation] \label{fact:truncation}
  Let $(s_0, \ldots, s_n)$ be an attainable tuple.
  Then $(s_0, \ldots, s_\ell)$ is attainable for every $1 \le \ell \le n$.
\end{fact}
\begin{proof}
  It suffices to show that $(s_0, \ldots, s_{n-1})$ is attainable.
  Write $s_k := s_k(y_1, \ldots, y_n)$ for some real numbers $y_1, \dots, y_n$.
  By Rolle's theorem, between every two consecutive real roots of a polynomial, there is a real root of its derivative.
  Thus the derivative of a real-rooted polynomial is also real-rooted.
  Therefore, the polynomial
  \begin{align*}
    \frac{1}{n} \cdot \frac{d}{dz} \sum_{k=0}^n (-1)^k \binom{n}{k} s_k(y) z^{n-k}
    &= \sum_{k=0}^{n-1} (-1)^k \frac{n-k}{n} \binom{n}{k} s_k(y) z^{n-1-k} \\
    &= \sum_{k=0}^{n-1} (-1)^k \binom{n-1}{k} s_k(y) z^{n-1-k}
  \end{align*}
  is monic and real-rooted, showing that $(s_0, \ldots, s_{n-1})$ is also attainable.
\end{proof}

\begin{remark}
  One should view $s_\ell$ as $s_\ell = \prod_{i=1}^\ell y'_i$ for some $y'_1, \ldots, y'_\ell \in \R$, instead of $s_\ell = \binom{n}{\ell}^{-1} \sum_{\abs{S}=\ell} y^S$ for some $y_1, \ldots, y_n \in \R$ such that $s_n = \prod_{i=1}^n y_i$.
\end{remark}

\Cref{lemma:elem-sym-ineq} relies on the following slight refinement in Tao's argument.

\begin{lemma} \label{lemma:tao}
  Let $(s_0, \ldots, s_n)$ be an attainable tuple.
  Then for every $1 \le \ell \le n$,
  \[
    \abs{s_\ell}^{\frac{2}{\ell}}
    \le (\ell-1) \cdot (s_1^2 - s_2) + s_1^2 .
  \]
\end{lemma}

\begin{proof}
  By the truncation property (\Cref{fact:truncation}), it suffices to consider the case $\ell=n$.
  Write $s_k := s_k(y_1, \ldots, y_n)$ for some $y = (y_1, \ldots y_n) \in \R^n$.
  By the AM-GM inequality, we have
  \[
    \abs{s_n(y)}^{\frac{2}{n}}
    = (y_1^2 \cdots y_n^2)^{\frac{1}{n}}
    \le \frac{1}{n} \sum_{i=1}^n y_i^2 .
  \]
  By the Newton identity we have
  \begin{align*}
    \sum_{i=1}^n y_i^2
    = \Bigl(\sum_{i=1}^n y_i \Bigr)^2 - 2 \sum_{1 \le i<j \le n} y_i y_j
    = n^2 s_1(y_1, \ldots, y_n)^2 - 2 \binom{n}{2} s_2(y_1, \ldots, y_n) .
  \end{align*}
  Therefore,
  \[
    \abs{s_n(y)}^{\frac{2}{n}}
    \le n s_1(y)^2 - (n-1)s_2(y)
    = (n-1) \bigl( s_1(y)^2 - s_2(y) \bigr) + s_1(y)^2 . \qedhere
  \]
\end{proof}

\Cref{lemma:elem-sym-ineq} immediately follows from \Cref{lemma:tao} by un-normalizing $s_\ell, s_1$ and $s_2$.

\begin{proof}[Proof of \Cref{lemma:elem-sym-ineq}] 
  Let $S_k(y) := \binom{n}{k} s_k(y) = \sum_{\abs{S}=k} y^S$.
  Applying \Cref{lemma:tao}, we have
  \begin{align*}
      \abs{S_\ell}^{\frac{2}{\ell}}
      &\le \binom{n}{\ell}^{\frac{2}{\ell}} \Bigl( (\ell-1) (s_1^2 - s_2) + s_1^2 \Bigr) \\
      &= \binom{n}{\ell}^{\frac{2}{\ell}} \biggl( (\ell-1) \Bigl( \frac{S_1^2}{n^2} - \frac{2S_2}{n(n-1)} \Bigr) + \frac{S_1^2}{n^2} \biggr) \\
      &= \binom{n}{\ell}^{\frac{2}{\ell}} \biggl( (\ell-1) \Bigl( \frac{S_1^2 - 2S_2}{n(n-1)} - \frac{S_1^2}{n^2(n-1)} \Bigr) + \frac{S_1^2}{n^2} \biggr) \\
      &= \binom{n}{\ell}^{\frac{2}{\ell}} \biggl( \frac{\ell-1}{n-1} \Bigl( \frac{S_1^2 - 2S_2}{n} - \frac{S_1^2}{n^2} \Bigr) + \frac{S_1^2}{n^2} \biggr) \\
      &= \binom{n}{\ell}^{\frac{2}{\ell}} \biggl( \frac{\ell-1}{n-1} \Bigl( \frac{S_1^2 - 2S_2}{n} \Bigr) + \left(1 - \frac{\ell-1}{n-1} \right) \frac{S_1^2}{n^2} \biggr) .
  \end{align*}
  Applying Newton's identity, i.e., $S_1^2 - 2S_2 = \sum_{i=1}^n y_i^2$, completes the proof.
\end{proof}

We now prove \Cref{prop:Harm-bound-H}.
We note that a similar argument also appears in \cite[Lemma 2.1]{Pol19}.
For completeness we provide a self-contained proof here.
\begin{proof}[Proof of \Cref{prop:Harm-bound-H}]
First note
$$
(1+z)^{(n+t)/2}(1-z)^{(n-t)/2}=\sum_{\ell=0}^n
\sK(\ell,t)z^\ell.
$$
Let $r=|z|$. The logarithmic function is known to be concave: 
\[
  \alpha\log_2(u)+(1-\alpha)\log_2(v)\leq \log_2(\alpha u+(1-\alpha) v).
\]
for any positive $u$ and $v$.
Using concavity and the observation   $|1+z|^2+|1-z|^2=2+2|z|^2=1+z+\overline{z}+|z|^2+1-z-\overline{z}+|z|^2=2+2|z|^2=2+2r^2$ gives 
\begin{align*}
  H(\alpha)+\log_2\big(|1+z|^{2\alpha}|1-z|^{2(1-\alpha)}\big)
  &= \alpha\log_2\Big(\frac{|1+z|^2}{\alpha}\Big) + (1-\alpha)\log_2\Big(\frac{|1-z|^2}{1-\alpha}\Big) \\
  &\le \log_2(|1+z|^2+|1-z|^2) \\
  &=\log_2(2+2r^2) .
\end{align*}
For an integer $\ell$ with $0\leq \ell\leq n$
consider the Laurent polynomial 
$$
p(z)=\frac{(1+z)^{(n+t)/2}(1-z)^{(n-t)/2}}{z^\ell}.
$$
If $r^2=\beta/(1-\beta)$, then we have
\begin{align*}
  \log|p(z)|
  &= \frac{n}{2}\Big(\log_2\big(|1+z|^{2\alpha}|1-z|^{2(1-\alpha)}\big)-\beta\log_2\big(|z|^2\big)\Big) \\
  &\leq \frac{n}{2}\Big(\log_2(2+2r^2)-\beta\log_2(r^2)-H(\alpha)\Big) \\
  &= \frac{n}{2}\Big(\log_2\tfrac{2}{1-\beta}-
\beta\log_2\tfrac{\beta}{1-\beta} - H(\alpha)\Big) \\
  &= \frac{n}{2}\Big(1+H(\beta)-H(\alpha)\Big) .
\end{align*}
The coefficient of $1=z^0$ in $p(z)$ is $\sK(\ell,t)$, so it follows that
$\sK(\ell,t)=\int_0^1p(re^{2\pi i \theta})\,d\theta$. We conclude that
\[
\log_2|\sK(\ell,t)|
\le \log_2\Big(\int_0^1 |p(re^{2\pi i\theta})|\,d\theta\Big)
\le \max_{|z|=r}\log_2 |p(z)|
\le\frac{n}{2}\Big(1-H(\alpha)+H(\beta)\Big). \qedhere
\]
\end{proof}

\subsection{Lower bound on Krawtchouk polynomials}

In this section, we prove \Cref{claim:sK-lb}.
It follows from an inequality on Krawtchouk polynomials which appears to be well known in the coding theory literature~\cite{MRRW77,KL95,Kra01,KS21}, and essentially follows from Newton's inequality.

For convenience we will work with the standard (non-shifted) definition of Krawtchouk polynomials $K(\ell,t) = \sK(\ell,n-2t)$.
Note that in the claim below, we intentionally swap $t$ and $\ell$.
\begin{claim} \label{claim:K-lb}
  $K(n/2-t,\ell) \ge \binom{n}{n/2-t} (t/n)^\ell$ for $t \ge \sqrt{\ell(n-\ell)}$.
\end{claim}

\Cref{claim:K-lb} follows from the fact that $K(t,0) = \binom{n}{t}$ and then applying the following lemma iteratively $\ell$ times.

\begin{lemma}[Theorem 8 in \cite{Kra01}]
  For $\ell,i$ such that $(n-2i)^2 \ge 4\ell(n-\ell)$ (so that $s = \sqrt{(n-2i)^2 - 4\ell(n-\ell)}$ is real and nonnegative),
  \[
    \frac{K(i,\ell+1)}{K(i,\ell)} > \frac{n-2i+s}{2(n-\ell)} \ge \frac{n-2i}{2n} .
  \]
\end{lemma}

We can now prove \Cref{claim:sK-lb} using the following fact and translating the statement in terms of $\sK(\cdot,\cdot)$.
\begin{fact}
  $\binom{n}{t} K(\ell,t) = \binom{n}{\ell} K(t,\ell)$ .
\end{fact}

\begin{proof}[Proof of \Cref{claim:sK-lb}]
  We have
  \[
    \sK(\ell,t)
    = K\Bigl(\ell,\frac{n}{2}-\frac{t}{2}\Bigr)
    = \frac{\binom{n}{\ell}}{\binom{n}{\frac{n}{2}+\frac{t}{2}}} K\Bigl(\frac{n}{2}-\frac{t}{2},\ell\Bigr)
    \ge \binom{n}{\ell} \Bigl(\frac{t}{2n}\Bigr)^\ell . \qedhere
  \]
\end{proof}

\paragraph{Acknowledgments.} CHL thanks Salil Vadhan and Terence Tao for helpful discussion.


\begin{thebibliography}{MRRW77}

\bibitem[AAK{\etalchar{+}}07]{AlonAKMRX07}
Noga Alon, Alexandr Andoni, Tali Kaufman, Kevin Matulef, Ronitt Rubinfeld, and Ning Xie.
\newblock Testing k-wise and almost k-wise independence.
\newblock In {\em ACM Symp.~on the Theory of Computing (STOC)}, pages 496--505, 2007.

\bibitem[ABI86]{ABI86}
Noga Alon, L{\'a}szl{\'o} Babai, and Alon Itai.
\newblock A fast and simple randomized algorithm for the maximal independent set problem.
\newblock {\em Journal of Algorithms}, 7:567--583, 1986.

\bibitem[AGM03]{AGM03}
Noga Alon, Oded Goldreich, and Yishay Mansour.
\newblock Almost {$k$}-wise independence versus {$k$}-wise independence.
\newblock {\em Inform. Process. Lett.}, 88(3):107--110, 2003.

\bibitem[Ahl17]{Ahle-book}
Thomas~D Ahle.
\newblock {\em Asymptotic Tail Bound and Applications}.
\newblock 2017.
\newblock Available at \url{https://thomasahle.com/papers/tails.pdf}.

\bibitem[Ajt83]{Ajt83}
Mikl{\'o}s Ajtai.
\newblock {$\Sigma \sp{1}\sb{1}$}-formulae on finite structures.
\newblock {\em Annals of Pure and Applied Logic}, 24(1):1--48, 1983.

\bibitem[AW89]{AjtaiW89}
Miklos Ajtai and Avi Wigderson.
\newblock Deterministic simulation of probabilistic constant-depth circuits.
\newblock {\em Advances in Computing Research - Randomness and Computation}, 5:199--223, 1989.

\bibitem[Baz07]{Bazzi07}
Louay Bazzi.
\newblock Polylogarithmic independence can fool {DNF} formulas.
\newblock In {\em 48th IEEE Symp.~on Foundations of Computer Science (FOCS)}, pages 63--73, 2007.

\bibitem[Baz15a]{Bazzi-pseudobinomial}
Louay Bazzi.
\newblock Entropy of weight distributions of small-bias spaces and pseudobinomiality.
\newblock In {\em Computing and combinatorics}, volume 9198 of {\em Lecture Notes in Comput. Sci.}, pages 495--506. Springer, Cham, 2015.

\bibitem[Baz15b]{Bazzi15}
Louay Bazzi.
\newblock Weight distribution of cosets of small codes with good dual properties.
\newblock {\em IEEE Trans. Inform. Theory}, 61(12):6493--6504, 2015.

\bibitem[BDVY13]{BDVY08}
Andrej Bogdanov, Zeev Dvir, Elad Verbin, and Amir Yehudayoff.
\newblock Pseudorandomness for width-2 branching programs.
\newblock {\em Theory Comput.}, 9:283--292, 2013.

\bibitem[BGGP12]{benjamini2012kwise}
Itai Benjamini, Ori Gurel-Gurevich, and Ron Peled.
\newblock On k-wise independent distributions and boolean functions, 2012.

\bibitem[BHLV19]{BHLV}
Ravi Boppana, Johan H{\aa}stad, Chin~Ho Lee, and Emanuele Viola.
\newblock Bounded independence versus symmetric tests.
\newblock {\em ACM Trans. Comput. Theory}, 11(4):Art. 21, 27, 2019.

\bibitem[BIJ{\etalchar{+}}21]{BIJLSV21}
Jaros{\l}aw B{\l}asiok, Peter Ivanov, Yaonan Jin, Chin~Ho Lee, Rocco~A. Servedio, and Emanuele Viola.
\newblock Fourier growth of structured {$\mathbb{F}_2$}-polynomials and applications.
\newblock In {\em Approximation, randomization, and combinatorial optimization. {A}lgorithms and techniques}, volume 207 of {\em LIPIcs. Leibniz Int. Proc. Inform.}, pages Art. No. 53, 20. Schloss Dagstuhl. Leibniz-Zent. Inform., Wadern, 2021.

\bibitem[Bra10]{Braverman10}
Mark Braverman.
\newblock Polylogarithmic independence fools {AC}$^{\mbox{0}}$ circuits.
\newblock {\em J.~of the ACM}, 57(5), 2010.

\bibitem[BS15]{BS15}
Mark Bun and Thomas Steinke.
\newblock Weighted polynomial approximations: limits for learning and pseudorandomness.
\newblock In {\em Approximation, randomization, and combinatorial optimization. {A}lgorithms and techniques}, volume~40 of {\em LIPIcs. Leibniz Int. Proc. Inform.}, pages 625--644. Schloss Dagstuhl. Leibniz-Zent. Inform., Wadern, 2015.

\bibitem[BV10]{BoV-gen}
Andrej Bogdanov and Emanuele Viola.
\newblock Pseudorandom bits for polynomials.
\newblock {\em SIAM J.~on Computing}, 39(6):2464--2486, 2010.

\bibitem[CGH{\etalchar{+}}85]{ChorGoHaFrRuSm85}
Benny Chor, Oded Goldreich, Johan H{\r{a}}stad, Joel Friedman, Steven Rudich, and Roman Smolensky.
\newblock The bit extraction problem or {t}-resilient functions (preliminary version).
\newblock In {\em 26th Symposium on Foundations of Computer Science}, pages 396--407, Portland, Oregon, 21--23 October 1985. IEEE.

\bibitem[CGL{\etalchar{+}}21]{CGLSS21}
Eshan Chattopadhyay, Jason Gaitonde, Chin~Ho Lee, Shachar Lovett, and Abhishek Shetty.
\newblock Fractional pseudorandom generators from any {F}ourier level.
\newblock In {\em 36th {C}omputational {C}omplexity {C}onference}, volume 200 of {\em LIPIcs. Leibniz Int. Proc. Inform.}, pages Art. No. 10, 24. Schloss Dagstuhl. Leibniz-Zent. Inform., Wadern, 2021.

\bibitem[CHHL19]{CHHL19}
Eshan Chattopadhyay, Pooya Hatami, Kaave Hosseini, and Shachar Lovett.
\newblock Pseudorandom generators from polarizing random walks.
\newblock {\em Theory Comput.}, 15:Paper No. 10, 26, 2019.

\bibitem[CHLT19]{CHLT19}
Eshan Chattopadhyay, Pooya Hatami, Shachar Lovett, and Avishay Tal.
\newblock Pseudorandom generators from the second {F}ourier level and applications to {AC}0 with parity gates.
\newblock In {\em 10th {I}nnovations in {T}heoretical {C}omputer {S}cience}, volume 124 of {\em LIPIcs. Leibniz Int. Proc. Inform.}, pages Art. No. 22, 15. Schloss Dagstuhl. Leibniz-Zent. Inform., Wadern, 2019.

\bibitem[CHRT18]{CHRT18}
Eshan Chattopadhyay, Pooya Hatami, Omer Reingold, and Avishay Tal.
\newblock Improved pseudorandomness for unordered branching programs through local monotonicity.
\newblock In {\em S{TOC}'18---{P}roceedings of the 50th {A}nnual {ACM} {SIGACT} {S}ymposium on {T}heory of {C}omputing}, pages 363--375. ACM, New York, 2018.

\bibitem[CLTW23]{CLTW23}
Lijie Chen, Xin Lyu, Avishay Tal, and Hongxun Wu.
\newblock New {PRG}s for unbounded-width/adaptive-order read-once branching programs.
\newblock In {\em 50th {I}nternational {C}olloquium on {A}utomata, {L}anguages, and {P}rogramming}, volume 261 of {\em LIPIcs. Leibniz Int. Proc. Inform.}, pages Art. No. 39, 20. Schloss Dagstuhl. Leibniz-Zent. Inform., Wadern, 2023.

\bibitem[CT06]{CT06}
Thomas~M. Cover and Joy~A. Thomas.
\newblock {\em Elements of information theory}.
\newblock Wiley-Interscience [John Wiley \& Sons], Hoboken, NJ, second edition, 2006.

\bibitem[CW79]{CaW79}
J.~Lawrence Carter and Mark~N. Wegman.
\newblock Universal classes of hash functions.
\newblock {\em J.~of Computer and System Sciences}, 18(2):143--154, 1979.

\bibitem[DGJ{\etalchar{+}}10]{DGJSV-bifh}
Ilias Diakonikolas, Parikshit Gopalan, Ragesh Jaiswal, Rocco~A. Servedio, and Emanuele Viola.
\newblock Bounded independence fools halfspaces.
\newblock {\em SIAM J.~on Computing}, 39(8):3441--3462, 2010.

\bibitem[DHH19]{DHH19}
Dean Doron, Pooya Hatami, and William~M. Hoza.
\newblock Near-optimal pseudorandom generators for constant-depth read-once formulas.
\newblock In {\em 34th {C}omputational {C}omplexity {C}onference}, volume 137 of {\em LIPIcs. Leibniz Int. Proc. Inform.}, pages Art. No. 16, 34. Schloss Dagstuhl. Leibniz-Zent. Inform., Wadern, 2019.

\bibitem[DHH20]{DBLP:conf/coco/DoronHH20}
Dean Doron, Pooya Hatami, and William~M. Hoza.
\newblock Log-seed pseudorandom generators via iterated restrictions.
\newblock In Shubhangi Saraf, editor, {\em 35th Computational Complexity Conference, {CCC} 2020, July 28-31, 2020, Saarbr{\"{u}}cken, Germany (Virtual Conference)}, volume 169 of {\em LIPIcs}, pages 6:1--6:36. Schloss Dagstuhl - Leibniz-Zentrum f{\"{u}}r Informatik, 2020.

\bibitem[DILV24]{SSS-II}
Harm Derksen, Peter Ivanov, Chin~Ho Lee, and Emanuele Viola.
\newblock Pseudorandomness, symmetry, smoothing: {II}.
\newblock 2024.

\bibitem[FK18]{ForbesKelley-2018}
Michael~A. Forbes and Zander Kelley.
\newblock Pseudorandom generators for read-once branching programs, in any order.
\newblock In {\em IEEE Symp.~on Foundations of Computer Science (FOCS)}, 2018.

\bibitem[GKM15]{GopalanKM15}
Parikshit Gopalan, Daniel Kane, and Raghu Meka.
\newblock Pseudorandomness via the discrete fourier transform.
\newblock In {\em IEEE Symp.~on Foundations of Computer Science (FOCS)}, pages 903--922, 2015.

\bibitem[GLS12]{GavinskyLS12}
Dmitry Gavinsky, Shachar Lovett, and Srikanth Srinivasan.
\newblock Pseudorandom generators for read-once acc{\^{}}0.
\newblock In {\em IEEE Conf.~on Computational Complexity (CCC)}, pages 287--297, 2012.

\bibitem[GMR{\etalchar{+}}12]{GopalanMRTV12}
Parikshit Gopalan, Raghu Meka, Omer Reingold, Luca Trevisan, and Salil Vadhan.
\newblock Better pseudorandom generators from milder pseudorandom restrictions.
\newblock In {\em IEEE Symp.~on Foundations of Computer Science (FOCS)}, 2012.

\bibitem[GY20]{GY20}
Parikshit Gopalan and Amir Yehudayoff.
\newblock Concentration for limited independence via inequalities for the elementary symmetric polynomials.
\newblock {\em Theory Comput.}, 16:Paper No. 17, 29, 2020.

\bibitem[HH23]{DBLP:journals/eccc/HatamiH23}
Pooya Hatami and William Hoza.
\newblock Theory of unconditional pseudorandom generators.
\newblock {\em Electron. Colloquium Comput. Complex.}, {TR23-019}, 2023.

\bibitem[HLV18]{HLV-bipnfp}
Elad Haramaty, Chin~Ho Lee, and Emanuele Viola.
\newblock Bounded independence plus noise fools products.
\newblock {\em SIAM J.~on Computing}, 47(2):295--615, 2018.

\bibitem[KL95]{KL95}
Gil Kalai and Nathan Linial.
\newblock On the distance distribution of codes.
\newblock {\em IEEE Trans. Inform. Theory}, 41(5):1467--1472, 1995.

\bibitem[KL01]{KL01}
Ilia Krasikov and Simon Litsyn.
\newblock Survey of binary {K}rawtchouk polynomials.
\newblock In {\em Codes and association schemes ({P}iscataway, {NJ}, 1999)}, volume~56 of {\em DIMACS Ser. Discrete Math. Theoret. Comput. Sci.}, pages 199--211. Amer. Math. Soc., Providence, RI, 2001.

\bibitem[Kra01]{Kra01}
Ilia Krasikov.
\newblock Nonnegative quadratic forms and bounds on orthogonal polynomials.
\newblock {\em J. Approx. Theory}, 111(1):31--49, 2001.

\bibitem[KS21]{KS21}
Naomi Kirshner and Alex Samorodnitsky.
\newblock A moment ratio bound for polynomials and some extremal properties of {K}rawchouk polynomials and {H}amming spheres.
\newblock {\em IEEE Trans. Inform. Theory}, 67(6, part 1):3509--3541, 2021.

\bibitem[Lee19]{Lee19}
Chin~Ho Lee.
\newblock Fourier bounds and pseudorandom generators for product tests.
\newblock 2019.

\bibitem[Lev95]{Levenshtein95}
Vladimir~I. Levenshtein.
\newblock Krawtchouk polynomials and universal bounds for codes and designs in {H}amming spaces.
\newblock {\em IEEE Trans. Inform. Theory}, 41(5):1303--1321, 1995.

\bibitem[LV17]{LV-sum}
Chin~Ho Lee and Emanuele Viola.
\newblock Some limitations of the sum of small-bias distributions.
\newblock {\em Theory of Computing}, 13, 2017.

\bibitem[LV20]{LeeV-rop}
Chin~Ho Lee and Emanuele Viola.
\newblock More on bounded independence plus noise: Pseudorandom generators for read-once polynomials.
\newblock {\em Theory of Computing}, 16:1--50, 2020.
\newblock Available at http://www.ccs.neu.edu/home/viola/.

\bibitem[MRRW77]{MRRW77}
Robert~J. McEliece, Eugene~R. Rodemich, Howard Rumsey, Jr., and Lloyd~R. Welch.
\newblock New upper bounds on the rate of a code via the {D}elsarte-{M}ac{W}illiams inequalities.
\newblock {\em IEEE Trans. Inform. Theory}, IT-23(2):157--166, 1977.

\bibitem[MRT19]{MRT19}
Raghu Meka, Omer Reingold, and Avishay Tal.
\newblock Pseudorandom generators for width-3 branching programs.
\newblock In {\em S{TOC}'19---{P}roceedings of the 51st {A}nnual {ACM} {SIGACT} {S}ymposium on {T}heory of {C}omputing}, pages 626--637. ACM, New York, 2019.

\bibitem[MZ09]{MZ09}
Raghu Meka and David Zuckerman.
\newblock Small-bias spaces for group products.
\newblock In {\em Approximation, randomization, and combinatorial optimization}, volume 5687 of {\em Lecture Notes in Comput. Sci.}, pages 658--672. Springer, Berlin, 2009.

\bibitem[Nis92]{Nis92}
Noam Nisan.
\newblock Pseudorandom generators for space-bounded computation.
\newblock {\em Combinatorica}, 12(4):449--461, 1992.

\bibitem[NN90]{NaN90}
J.~Naor and M.~Naor.
\newblock Small-bias probability spaces: efficient constructions and applications.
\newblock In {\em 22nd ACM Symp.~on the Theory of Computing (STOC)}, pages 213--223. ACM, 1990.

\bibitem[O'D14]{ODonnell14}
Ryan O'Donnell.
\newblock {\em Analysis of Boolean Functions}.
\newblock Cambridge University Press, 2014.

\bibitem[OZ18]{OZ18}
Ryan O'Donnell and Yu~Zhao.
\newblock {On Closeness to $k$-Wise Uniformity}.
\newblock In Eric Blais, Klaus Jansen, Jos{\'e} D.~P. Rolim, and David Steurer, editors, {\em Approximation, Randomization, and Combinatorial Optimization. Algorithms and Techniques (APPROX/RANDOM 2018)}, volume 116 of {\em Leibniz International Proceedings in Informatics (LIPIcs)}, pages 54:1--54:19, Dagstuhl, Germany, 2018. Schloss Dagstuhl--Leibniz-Zentrum fuer Informatik.

\bibitem[Pol19]{Pol19}
Yury Polyanskiy.
\newblock Hypercontractivity of spherical averages in {H}amming space.
\newblock {\em SIAM J. Discrete Math.}, 33(2):731--754, 2019.

\bibitem[Raz09]{Razborov09}
Alexander~A. Razborov.
\newblock A simple proof of {B}azzi's theorem.
\newblock {\em ACM Transactions on Computation Theory (TOCT)}, 1(1), 2009.

\bibitem[RR47]{MR22821}
C.~Radhakrishna~Rao.
\newblock Factorial experiments derivable from combinatorial arrangements of arrays.
\newblock {\em Suppl. J. Roy. Statist. Soc.}, 9:128--139, 1947.

\bibitem[RSV13]{ReingoldSV13}
Omer Reingold, Thomas Steinke, and Salil~P. Vadhan.
\newblock Pseudorandomness for regular branching programs via {F}ourier analysis.
\newblock In {\em Workshop on Randomization and Computation (RANDOM)}, pages 655--670, 2013.

\bibitem[SKS19]{DBLP:conf/focs/SilbakKS19}
Jad Silbak, Swastik Kopparty, and Ronen Shaltiel.
\newblock Quasilinear time list-decodable codes for space bounded channels.
\newblock In David Zuckerman, editor, {\em 60th {IEEE} Annual Symposium on Foundations of Computer Science, {FOCS} 2019, Baltimore, Maryland, USA, November 9-12, 2019}, pages 302--333. {IEEE} Computer Society, 2019.

\bibitem[SVW17]{SVW17}
Thomas Steinke, Salil Vadhan, and Andrew Wan.
\newblock Pseudorandomness and {F}ourier-growth bounds for width-3 branching programs.
\newblock {\em Theory Comput.}, 13:Paper No. 12, 50, 2017.

\bibitem[Tal17]{Tal17}
Avishay Tal.
\newblock Tight bounds on the fourier spectrum of {AC0}.
\newblock In {\em Conf.~on Computational Complexity (CCC)}, pages 15:1--15:31, 2017.

\bibitem[Tao23a]{Tao23-private}
Terence Tao.
\newblock Personal communication, 2023.

\bibitem[Tao23b]{Tao23-Maclaurin}
Terence Tao.
\newblock {A Maclaurin type inequality}, 2023.
\newblock Available at \url{https://arxiv.org/abs/2310.05328}.

\bibitem[Vad12]{Vadhan12}
Salil~P. Vadhan.
\newblock Pseudorandomness.
\newblock {\em Foundations and Trends in Theoretical Computer Science}, 7(1-3):1--336, 2012.

\bibitem[Vio09]{ViolaBPvsE}
Emanuele Viola.
\newblock On approximate majority and probabilistic time.
\newblock {\em Computational Complexity}, 18(3):337--375, 2009.

\bibitem[Vio22a]{corr-survey}
Emanuele Viola.
\newblock Correlation bounds against polynomials, a survey.
\newblock 2022.

\bibitem[Vio22b]{viola-tm}
Emanuele Viola.
\newblock Pseudorandom bits and lower bounds for randomized turing machines.
\newblock {\em Theory of Computing}, 18(10):1--12, 2022.

\bibitem[Vio23]{moti}
Emanuele Viola.
\newblock Mathematics of the impossible: The uncharted complexity of computation.
\newblock 2023.

\end{thebibliography}
\newcommand{\etalchar}[1]{$^{#1}$}
\def\cprime{$'$}

\end{document}